%% file: Symplectic_Lattices_from_Crypto.tex
\documentclass{article}
\input{commands.tex}

\title{Symplectic Lattices and GKP Codes - Simple Randomized Constructions from Cryptographic Lattices}
\author[]{Johannes Bl\"omer}
\author[]{Yinzi Xiao}
\author[]{Zahra Raissi}
\author[]{Stanislaw Soltan}
\date{}

\affil[]{\small{Computer Science Department, Paderborn University, Paderborn, Germany\\
Institute for Photonic System (PhoQS), Paderborn, Germany}} 

\begin{document}
\maketitle

\begin{abstract}
We construct good GKP (Gottesman-Kitaev-Preskill) codes (in the sense of Conrad, Eisert and Seifert proposed) from standard \emph{short integer solution ($\sis$)} lattices as well as from ring $\sis$ and module $\sis$ lattices, $\rsis$ and $\msis$ lattices, respectively. These lattices are crucial for lattice-based cryptography. Our construction yields GKP codes with distance $\sqrt{n/\lambda\pi e}$ for $\lambda^n$ logical dimensions. This compares favorably with the NTRU-based construction by Conrad et al. that achieves distance $\Omega(\sqrt{n/q}),$ with $n\le q^2/0.28$. Unlike their codes, our codes do not have secret keys that can be used to speed-up the decoding. However, combining our results based on $\rsis$ with results by Stehl\'{e} and Steinfeld, we can prove Conjecture 2 from Conrad et al., namely with appropriate parameter choices, NTRU lattices yield GKP codes with distance $\sqrt{n/\pi e}$. Finally, we present a simple decoding algorithm that, for many parameter choices, experimentally yields decoding results similar to the ones for NTRU-based codes. Using the $\rsis$ and $\msis$ construction, our simple decoding algorithm runs in nearly linear time. Following Conrad, Eisert and Seifert's work, our construction of GKP codes follows directly from an explicit, randomized construction of symplectic lattices with (up to constants $\approx 1$) minimal distance $(1/\sigma_{2n})^{1/2n}\approx \sqrt{\frac{n}{\pi e}}$, where $\sigma_{2n}$ is the volume of the $2n$-dimensional unit ball. Before this result, Buser and Sarnak gave a non-constructive proof for the existence of such symplectic lattices. 
\end{abstract}

\input{Introduction}
\input{Basics}
\input{Constructions}
\input{Analysis}
\input{MainTheorems}
\input{Bounded_Distance_Decoding}
\input{Simulation}

\input{Ring_Constructions.tex}
\input{Ring_Analysis.tex}
\input{MainTheoremsRings}
\input{Ring_and_NTRU}
\input{Ring_BDD}
\input{Ring_Simulation}
\input{Conclusion}

\input{Acknowledgement}

\printbibliography

\end{document}

%% file: commands.tex
\usepackage[english]{babel}

\usepackage{amsmath, amssymb, amsthm}
\usepackage{graphicx}
\usepackage[colorlinks=true, allcolors=blue]{hyperref}
\usepackage{todonotes}
\usepackage{subcaption}
\usepackage{authblk}

\usepackage[style=numeric,sorting=none,citestyle=numeric-comp,backend=biber,doi=false,url=false]{biblatex}
\newbibmacro{linkdoiurl}[1]{%
  \iffieldundef{doi}
    {\iffieldundef{url}
       {#1}
       {\href{\thefield{url}}{#1}}}
    {\href{https://doi.org/\thefield{doi}}{#1}}
}
\DeclareBibliographyDriver{article}{%
  \usebibmacro{author/translator+others}%
  \setunit{\labelnamepunct}\newblock
  \usebibmacro{title}%
  \newunit
  \usebibmacro{journal+issuetitle}
  \newunit
  \usebibmacro{doi+eprint+url}%
  \newunit
  \usebibmacro{finentry}%
}
\renewbibmacro*{journal+issuetitle}{%
  \usebibmacro{linkdoiurl}{%
    \printfield[journaltitle]{journaltitle}%
    \setunit*{\addspace}%
    \printfield[volume]{volume}%
    \setunit{\addcomma\space}%
    \printfield{year}%
    \setunit*{\addcomma\space}%
    \printfield{pages}%
  }%
}

\newtheorem{theorem}{Theorem}
\newtheorem{lemma}[theorem]{Lemma}
\newtheorem{corollary}[theorem]{Corollary}
\newtheorem{definition}[theorem]{Definition}

\newcommand{\Z}{\mathbb{Z}}
\newcommand{\N}{\mathbb{N}}
\newcommand{\R}{\mathbb{R}}
\newcommand{\Q}{\mathbb{Q}}

\newcommand{\F}{\mathbb{F}}

\newcommand{\sis}{\textrm{SIS}}

\newcommand{\bdd}{\textrm{BDDx}}

\newcommand{\rsis}{\textrm{R-SIS}}

\newcommand{\msis}{\textrm{M-SIS}}

\newcommand{\lwe}{\textrm{LWE}}

\newcommand{\ring}{\Z[X]}
\newcommand{\ringq}{\Z_q[X]}

\newcommand{\rring}{\ring/(X^n+1)}
\newcommand{\rringq}{\ringq/(X^n+1)}

\newcommand{\fq}[2]{\F_{#1^{#2}}}
\newcommand{\Fq}[2][]{\F_{#2}^{#1}}

\newcommand{\cL}{{\cal L}}

\newcommand{\cI}{{\cal I}}

\newcommand{\cT}{{\cal T}}

\newcommand{\cC}{{\cal C}}

\newcommand{\cO}{{\textbf{{\cal O}}}}

\newcommand{\latperp}{\cL^\perp}
\newcommand{\latdag}{\cL^\dagger}
\newcommand{\latcross}{\cL^\times}

\newcommand{\radius}{r\sqrt{q}}

\newcommand{\dradius}{r\sqrt{qd/n}}

\newcommand{\vb}[1]{\textbf{#1}}

\newcommand{\matsym}[2]{\sym(#1,#2)}

\DeclareMathOperator{\sym}{Sym}

\newcommand{\rmatsym}[1]{R_{q,#1}^{\sym}}

\DeclareMathOperator{\lcm}{lcm}
\DeclareMathOperator{\spann}{span}
\DeclareMathOperator{\dist}{dist}

\newcommand{\zqn}{\Z_q^{n\times n}}

\newcommand{\qfactor}{\frac{1}{\sqrt{q}}}

\newcommand{\vbh}{\vb{z}_{ij}^{(h)}}
\newcommand{\vbhu}{\overline{\vb{z}}_{ij}^{(h)}}
\newcommand{\vbhl}{\underline{\vb{z}}_{ij}^{(h)}}

\newcommand{\bddtriv}{\textsc{Bdd}_{\text{triv}}}

\newcommand{\setn}[1]{[#1]}
\newcommand{\setT}[2]{T_{#1}(#2)}
\newcommand{\csetT}[2]{\bar{T}_{#1}(#2)}

\newcommand{\ntru}{\text{NTRU}}
\newcommand{\distntru}[1]{R_{#1}^{\ntru}}

%% file: Introduction.tex
\section{Introduction}\label{sec:intro}

Quantum error correction and fault tolerant quantum computing are essential ingredients for the realization of fully fledged quantum computing in the future. Which quantum error correction codes can be realized efficiently determines the kind of physical platforms most promising in building quantum computers; they influence the complete quantum computing stack, from the physical qubit level to the quantum software level. 
Bosonic quantum error-correcting codes encode discrete logical information in infinite-dimensional systems and promise hardware-efficient routes to fault tolerance on continuous-variable platforms, e.\ g.\ photonic-based computing platforms \cite{terhal2020towards,Terhal-Review-2015,Glancy-Knill-2006,MenicucciPhysRevLett-2014}. A leading example are the Gottesman–Kitaev–Preskill (GKP) codes, stabilizer error correcting codes which place stabilizers on a regular grid in phase space so that small displacement errors can be detected and corrected. This idea originates with the seminal paper~\cite{gottesman2001encoding} by Gottesman, Kitaev and Preskill. Experimentally, grid encodings and error-correction primitives have been realized in superconducting cavities and trapped ions, demonstrating preparation, syndrome extraction, and logical protection for oscillator-encoded qubits \cite{Campagne_Ibarcq_2020-exp,fluhmann2019encoding}.
More recently, an integrated photonic platform has produced optical GKP states on-chip with multiple resolvable peaks and negative-Wigner structure, underscoring growing practical momentum \cite{Larsen2025}.\\

Lattices from the geometry of numbers (see~\cite{cassels1996introduction,gruber1987geometry,conway2013sphere})  provide a convenient way to formalize GKP codes. In this formalism, an $n$-mode GKP  code is specified by a $2n$-dimensional symplectic lattice whose vectors generate the displacement stabilizers \cite{conrad2022gottesman,PRXQuantum-Grid-States}.
Key code parameters become geometric: in particular, the code distance—the minimum displacement magnitude that causes a logical error—is proportional to the length of shortest nonzero vector in the symplectic lattice. This lattice viewpoint also interfaces directly with channel coding, symplectically integral lattices in $2n$-dimensional phase space underlie the best-known achievable rates for the Gaussian displacement channel, motivating lattice codes spanning many oscillators \cite{Harrington-Preskill_2001,Wu_2023,Mao-Lin-PRXQuantum-2023,brock2025quantum}.\\

On the theory side, it is natural to call a family of  $n$-mode codes  \emph{good} when the distance of the codes in the family grows asymptotically at least as $\Omega(\sqrt{n})$, reflecting typical high-dimensional lattice behavior where shortest-vector lengths scale as $\Omega(\sqrt{n})$ (for lattices with determinant $1$ and dimension $n$), while having asymptotically non-vanishing rate (the ratio of encoded logical qubits to physical modes is bounded away from zero) (see \cite{conrad2024good} and the references therein). Based on a non-constructive existence proof of symplectic lattices with shortest vector length $\Omega(\sqrt{n})$ by Buser and Sarnak~\cite{buser1994period}, Harrington and Preskill~\cite{Harrington-Preskill_2001} established the existence of families of good GKP codes. Complementary bounds based on continuous-variable MacWilliams identities indicate that, in low dimensions, the strongest GKP distances come from canonical "best-packing" lattices— $E_8$ (8D) and the Leech lattice (24D)—widely known for maximizing minimum distance in those dimensions \cite{burchards2025qwilliamsidentities}.\\

Constructively, Conrad, Eisert, and Seifert~\cite{conrad2024good} showed that multi-mode GKP codes can be instantiated from NTRU-type lattices. NTRU is a popular lattice-based public-key cryptosystem introduced in the 1990s by Hoffstein et al.\ \cite{hoffstein1999ntru}.  The construction by Conrad et al.\  yields randomized GKP code families with non-vanishing rate  and, with probability exponentially close to $1$, code distance scaling  like $\Theta(\sqrt{n/q})$, where the probability is over the random choice of lattices and $q$ is some parameter depending on $n$ satisfying $n\le q^2/0.28$. Exploiting so-called lattice trapdoors for NTRU lattices, Conrad et al.\ also describe decoding algorithms for their GKP codes that beat generic decoding procedure based on general algorithms solving the closest vector problem in lattices, such as Babai's algorithm~\cite{babai1986lovasz}.\\ 

Note that due to the constraint $n\le q^2/0.28$, the distance bound $\Theta(\sqrt{n/q})$ scales like $\Theta(\sqrt[4]{n})$ rather than $\Theta(\sqrt{n})$ as required for a good family of GKP codes. Based on the so called Gaussian Heuristic (GH) (see \cite{conrad2024good, aono2019random,silverman2020introduction}) and experiments in low dimensional experiments, Conrad et al.\ conjecture that variants of their NTRU construction lead to families of GKP codes with $\lambda^n$ logical dimensions
 using $n$ bosonic modes and achieve code distance $\Delta\ge \sqrt{\frac{n}{\lambda \pi e}}$ (Conjecture 1 and 2 in \cite{conrad2024good}). They also provide arguments showing that replacing NTRU lattices by more general lattices based on $q$-ary symmetric matrices leads to GKP codes with distance $\Delta\ge \sqrt{\frac{n}{\lambda \pi e}}$  (see the introduction and Proposition 3 in \cite{conrad2024good}).\\    

In this paper, we present a simple randomized construction of multi-mode GKP codes from lattices as they appear in the (standard) \emph{short integer solutions ($\sis$)} problem and its variants, the \emph{ring short integer solutions ($\rsis$)} problem and the \emph{module short integer solutions ($\msis$)} problem.  These lattices have been used to define e.\ g.\ collision-resistant hash functions and unforgeable signature schemes in cryptography (for an introduction see~\cite{peikert2016decade} although we do not use any crypto-relevant properties of these lattices). 
Following basically the (randomized) construction blueprint provided by \cite{conrad2024good} and using $\sis,\rsis$ and $\msis$ lattices, we obtain good families of GKP codes that achieve code distance $\approx \sqrt{\frac{n}{\lambda\pi e}}$ for $n\log(\lambda)$ encoded logical qubits, that is, exactly the distances conjectured by Conrad et al.\ for variants of NTRU lattices. Our code families achieve these distances with probability close to $1$, where the probability is over a uniform random choice of a lattice from some finite class of lattices with dimension $n$, the set of symmetric $n$-dimensional $\sis$ or $\msis$ lattices to be precise. Choosing a lattice uniformly at random from these classes can be done efficiently, leading to an efficient construction of good GKP codes. The construction  based on $\sis$ lattices yields the most general class of GKP codes. It is also the easiest to analyze. We also provide constructions based on $\msis$ lattices ($\rsis$ lattices are just a special case) for two reasons.
\begin{enumerate}
    \item our main result on GKP codes based on $\rsis$ can be combined with results from Stehl\'{e} and Steinfeld~\cite{stehle2011making} to prove Conjecture 2 by Conrad, Eisert and Seifert~\cite{conrad2024good},  
    \item more importantly, we also provide simple decoding algorithms for our GKP codes; using standard methods from lattice-based cryptography  (e.\ g.\ fast Fourier transform), for $\msis$ lattices the algorithms can be made to run in almost linear time, whereas for $\sis$ based codes they require quadratic time.
\end{enumerate}
We obtain good GKP codes by first constructing so called \textit{$q$-symplectic} lattices with minimal distance $\approx\sqrt{\frac{qn}{\pi e}}$ (with high probability) from  $\sis$ lattices defined via \textit{symmetric matrices}\footnote{These lattices are also used by Conrad et al.\ for Proposition 3 to justify Conjecture 1 and 2. It is the interpretation of these lattices as $\sis$ lattices that leads to our main results.}. Here $q$ is a parameter that defines an appropriate finite field. Scaling these lattices by $\frac{1}{\sqrt{q}}$ we obtain \textit{symplectic} lattices with distance $\approx \sqrt{\frac{n}{\pi e}}$. Further scaling these lattices by $\sqrt{\lambda}$ yields GKP codes with the desired distance and with $n\log(\lambda)$ logical qubits. Note that symplecticity of the lattices is needed to obtain GKP codes. As a by-product of our GKP code construction, we obtain a probabilistic construction of $n$-dimensional symplectic lattices  with distance $\sqrt{\frac{n}{\pi e}}$. To the best of our knowledge, previously,  the existence of such symplectic lattices was only proven non-constructively (see~\cite{buser1994period})\footnote{In~\cite{buser1994period}, as well as in many other papers, results of this type are described in terms of $(1/\sigma_{2n})^{1/2n}$, where $\sigma_{2n}$ denotes the volume of the $2n$-dimensional unit ball. By Stirling's approximation $(1/\sigma_{2n})^{1/2n}\approx \sqrt{\frac{n}{\pi e}}$.}.\\ 

In addition to the blueprint from \cite{conrad2024good}, we also use standard techniques to analyze the minimal distance of random $\sis$ lattices. These analyses (for $\sis$ and $\msis$ lattices) are our main technical contribution. For $\sis$ lattices, the analysis of the minimal distance closely follows~\cite{micciancio2011geometry}. However, as mentioned above, to obtain GKP codes we cannot choose arbitrary $\sis$ lattices. Instead, their defining matrices must be symmetric. It turns out that this poses no problem, as we prove that random symmetric matrices over finite fields share an important regularity property with general matrices: for any fixed vector $\vb{v}$, its image under a uniformly random matrix (general or symmetric) is also uniformly distributed. For $\msis$ lattices, the symmetry issue also exists and can be resolved as for $\sis$ lattices. However, for $\msis$ lattices a more difficult technical problem arises since we have to deal with equations over finite rings rather than finite fields. More precisely, for $\msis$ lattices we deal with the ring $\Z_q[X]/(X^n+1)$, where $\Z_q=\Z/(q\Z)$. Although we present a very general analysis, it leads to non-trivial results only for specific choices of $n$ and $q$. The easiest and probably most useful case occurs when $n$ is a power of $2$ and $q=3,5\mod 16$, in which case the polynomial $X^n+1$ factors over $\Z_q[X]$ into two irreducible polynomials with the same degree. Note that these choices for $n,q$ are also popular parameter choices in cryptography and they correspond to useful choices for the NTRU-based constructions in \cite{conrad2024good}. Our techniques are similar to the ones used by Stehl\'{e} and Steinfeld~\cite{stehle2011making}. However, they focus on the $\ell_\infty$-norm, whereas we consider only the Euclidean norm. Moreover, since we are interested in exact constants we have to resolve additional problems stemming from the ring structure of $\Z_q[X]/(X^n+1)$. In particular, when proving the fundamental Lemma~\ref{lem:general_bound_set_size} we apply an averaging argument to deal with zero-divisors in $\Z_q[X]/(X^n+1)$. This argument is specific to the Euclidean norm.\\

Using standard tools from lattice theory one can show that GKP codes from symplectic lattices cannot achieve distance better than $\sqrt{2n/\lambda},$ where $n$ as before is the number of modes and $n\log(\lambda)$ is the number of encoded logical qubits. This bound follows easily from so-called transference bounds~\cite{banaszczyk1993new} and similar bounds were also shown in~\cite{conrad2022gottesman}. This bound shows that our construction achieves, up to small constants, the best possible distance. The \emph{Gaussian Heuristic} (see~\cite{conrad2024good}) indicates that the distance of our codes is even best possible, at least as long as one insists on achieving the distance with probability close to $1$.\\

Finally, using  the FPLLL  library~\cite{fpylll} for lattice reduction,  we provide various numerical simulations for our constructions of GKP codes and our simple decoding procedure. These simulations not only confirm the theoretical findings, but also show that at least in small dimensions, more general parameter choices than supported by our theoretical results lead to GKP codes with distance $\ge \sqrt{\frac{n}{\lambda\pi e}}$. To give an example, for $\sis$ lattices to yield good GKP codes with probability close to $1$, the theory requires $q\approx n^2$. The experiments show that we can choose $q$ much smaller than $n^2$ and still obtain codes with large distance. Although $q$ does not show up in our bounds for the code distance (it only influences probabilities), it does play an important role in our decoding algorithm. Here smaller values of $q$ lead to significantly better decoding properties. In fact, by choosing small $q$'s, our decoding procedure compares well to the NTRU decoding algorithm by Conrad et al.\cite{conrad2024good}. Similarly, our experiments show that the $\msis$-based construction is much less sensitive to choices for $n$ and $q$
than can be deduced from the theoretical results.\\

\paragraph{Organization} In Section~\ref{sec:basic} we provide all necessary concepts and results from lattice theory. The $\sis$-based construction of symplectic lattices is given in Section~\ref{sec:construction}, the analysis is presented in Section~\ref{sec:analysis_distance}, and the application to the construction of good GKP codes in Section~\ref{sec:main_results}. Here we also provide some well-known bounds on the invariants of symplectic lattices. In Section~\ref{sec:bdd} we present our simple decoding algorithm. Section~\ref{sec:simulation} summarizes our numerical findings for the $\sis$-based construction of GKP codes. In Sections~\ref{sec:ring_symplectic_construction} through~\ref{sec:ring_simulation} we show how to construct symplectic lattices from $\msis$ lattices (Section~\ref{sec:ring_symplectic_construction}), we analyze the construction (Section~\ref{sec:ring_analysis_distance}), apply it to the construction of GKP codes (Section~\ref{sec:maintheoremsring}), prove Conjecture 2 from~\cite{conrad2024good} (Section~\ref{sec:ntru_gkp}), describe the adaptation of our decoding algorithm to the $\msis$-based construction (Section~\ref{sec:ring_bdd}),  and present and discuss the experimental results for the $\msis$-based construction (Section~\ref{sec:ring_simulation}). \\

%% file: Basics.tex
\section{Basic definitions and results}\label{sec:basic}
Throughout this paper, vectors in some $\R^d$ are denoted by bold  letters. A vector $\vb{v}$ is always understood as a row vector.  Similarly, matrices will be denoted by capitalized bold letters. For a vector $\vb{v}\in \R^d$, we denote its standard euclidean distance by $\|\vb{v}\|$. 
\paragraph{Lattices}
A \emph{lattice} $\cL \subset \R^d$ is a discrete, abelian subgroup of the euclidean space $\R^d$. Every lattice $\cL\subset \R^d$ has a \emph{basis} consisting of $k,k\le d,$ $\R$- linearly independent vector $\vb{b}_1,\ldots, \vb{b}_k$, i.e., as a set,
\[
\cL:=\left\{\sum_{i=1}^k z_i \vb{b}_i:z_i\in \Z, i=1,\ldots,k\right\}.
\]
$k$ is called the \emph{dimension} of $\cL$. Collecting the basis vectors $\vb{b}_i$ in the rows of matrix $\vb{B}\in \R^{k\times d}$, we write 
\[
\cL=\cL(\vb{B})=\Z^k \vb{B}.
\]
We call a lattice $\cL=\cL(\vb{B}), \vb{B}\in \R^{k\times d}$ \emph{full-dimensional} iff $k=d$.

Given some lattice $\cL\subset \R^d$ its \emph{minimal distance} or its \emph{first successive minimum $\lambda_1(\cL)$} is defined as 
\[
\lambda_1(\cL):=\min\{\|\vb{v}\|:\vb{v}\in \cL, \vb{v}\ne \vb{0}\}.
\]
Because of the discrete structure of a lattice, this minimum always exists (and is in fact a minimum, not just an infimum). The \emph{covering radius $\rho(\cL)$} of lattice $\cL$ is defined as
\[
\rho(\cL):=\max\{\dist(\vb{t},\cL) \mid \vb{t}\in \spann(\cL)\},
\]
where for $\vb{t}\in \spann(\cL)$
\[
\dist(\vb{t},\cL):=\min \{\|\vb{t}-\vb{u} \| \mid \vb{u}\in \cL\}.
\]
Geometrically, $\rho(\cL)$ is the smallest radius such that spheres with this radius centered at every lattice point cover the entire space $\spann(\cL)$.
The (canonical or standard) \emph{dual lattice $\latdag$} of lattice $\cL$ is defined as 
\[
\latdag:=	\left\{ \vb{v}\in \spann(\cL)\mid \forall \vb{u}\in \cL: \vb{v}\cdot \vb{u}^T\in \Z\right\}.
\]
If $\cL=\cL(\vb{B})$ is full-dimensional, then $\latdag=\cL(\vb{B}^{-T})$.

The following celebrated theorem is due to Banszczyk. It is the strongest possible form (up to constants) of a \emph{transference theorem} relating invariants of a lattice and its dual (see~\cite{banaszczyk1993new}) .
\begin{theorem}[Banaszczyk]\label{thm:transference}
	Let $\cL$ be a lattice of dimension $n$. Then for lattice $\cL$ and its dual $\latdag$,
\[
\lambda_1(\latdag)\cdot \rho(\cL)\le \frac{n}{2}.	
\]
\end{theorem}
\paragraph{Symplectic matrices and symplectic lattices}
Set 
 \[
\vb{J}_{2n}:=\left(
\begin{matrix}
 \vb{0}_n & \vb{I}_n \\
 -\vb{I}_n & \vb{0}_n
\end{matrix}
\right).
\]
where $\vb{0}_n$ denotes the $n\times n$ matrix with all entries being $0$ and $\vb{I}_n$ denotes the $n\times n$ identity matrix. This matrix defines the standard symplectic form on $\R^{2n}$. Let $q\in\N$. A matrix  $\vb{M}\in \Z^{2n\times 2n}$ is called \textit{$q$-symplectic} if
\[
\vb{M}\cdot \vb{J}_{2n} \cdot \vb{M}^T=q\vb{J}_{2n}.
\]
A $1$-symplectic matrix is called \emph{symplectic.}  We call a $2n$-dimensional lattice  $\cL$ \textit{$q$-symplectic}, if it has a basis $\vb{M}_1, \vb{M}_2,\ldots, \vb{M}_{2n}$ such that the matrix $\vb{M}$ with rows $\vb{M}_1,\ldots, \vb{M}_{2n}$ is $q$-symplectic. A $1$-symplectic lattice is called \emph{symplectic}. The following connection between symplectic and $q$-symplectic matrices and lattices is straightforward.
\begin{lemma}\label{lem:q-symplectic_symplectic}
	Let $\vb{M}\in \Z^{n\times n}$ be a $q$-symplectic matrix. Then $\frac{1}{\sqrt{q}}\vb{M}$ is symplectic. Accordingly, for a $q$-symplectic matrix $\vb{M}$, lattice $\cL(\frac{1}{\sqrt{q}}\vb{M})$ is symplectic.
\end{lemma}
\begin{definition}\label{def:symplectic_dual}
	Let $\cL$ be a $2n$-dimensional lattice. Then the symplectic dual $\latcross$ of $\cL$ is defined as
	\[
	\latcross:=\left\{ \vb{v}\in \spann(\cL)\mid \forall \vb{u}\in \cL: \vb{u}\cdot \vb{J}_{2n} \cdot \vb{v}^T\in \Z\right\},
	\]
\end{definition}
This is analogous to the canonical dual lattice, but instead of the euclidean inner product uses the symplectic inner product $\langle \vb{u},\vb{v}\rangle_{\text{symp}}=\vb{u}\vb{J}_{2n}\vb{v}^T$.
The following two lemmata are well-known and easy to verify.
\begin{lemma}\label{lem:symplectic_self_dual}
	Let $\cL$ be a symplectic lattice. Then $\cL$ is symplectically self-dual, i.e.\ $\cL=\latcross$. Furthermore, for any $\alpha\in \R^{>0}$,
	\begin{equation}\label{eq:mult_symplectic}
		(\alpha\cL)^\times =\frac{1}{\alpha}\cL.
	\end{equation}
\end{lemma}
\begin{lemma}\label{lem:symplectic_canonical}
	Let $\cL$ be a full-dimensional lattice with dual $\latdag$ and symplectic dual $\latcross$. Then $\latcross=\latdag\cdot \cO$ for some orthonormal linear transformation $\cO$.
\end{lemma}
\begin{proof}
	Assume that $\cL=\cL(\vb{B}), \vb{B}\in \R^{2n\times 2n}$,, then $\latcross$ can be generated by the rows of matrix $\vb{B}^\times$ satisfying $\vb{B}^\times \cdot \vb{J}_{2n}\cdot \vb{B}^T=\vb{I}_{2n}$. Hence $\vb{B}^\times = \vb{B}^{-T}\cdot \vb{J}_{2n}^T$. The lemma follows, since $\vb{J}_{2n}$ is orthonormal.
\end{proof}
By this lemma, the canonical and symplectic dual of a lattice $\cL$ have the same invariants like the minimal distance or the covering radius. Combining the two previous lemmata with Banaszczyk's transference theorem (Theorem~\ref{thm:transference}) we obtain 
\begin{theorem}\label{thm:symplectic_covering_radius}
	Let $\cL$ be a symplectic lattice and $\alpha\in\R^+$. Then
	\[
	\rho((\alpha \cL)^\times)=\frac{1}{\alpha}\rho(\cL)\le \frac{n}{2\alpha\lambda_1(\cL)}.
	\]
\end{theorem}
\begin{proof}
The theorem follows from the following sequence of (in-)equalities.
\begin{align*}
	\rho((\alpha\cL)^\times) = \frac{1}{\alpha}\rho(\latcross)  & =  \frac{1}{\alpha}\rho(\cL)\\
	 & \le \frac{1}{\alpha}\frac{n}{2\lambda_1(\latdag)}\\
	 & = \frac{1}{\alpha}\frac{n}{2\lambda_1(\latcross)}\\
	 & = \frac{1}{\alpha}\frac{n}{2\lambda_1(\cL)}. 
\end{align*}
The first equality follows from~(\ref{eq:mult_symplectic}), the second and last equality follow from $\cL$ being symplectic. The inequality follows from the transference theorem. The third equality follows from Lemma~\ref{lem:symplectic_canonical}.
\end{proof}
\paragraph{A simple geometric lemma}
 By $B_d(r)$ denote the $d$-dimensional ball of radius $r$ centered at the origin.  The volume of $B_d(1)$ is denoted by $\sigma_d$. The next lemma is well-known. Below, it plays a crucial role in the probabilistic analysis of the minimum distance of random symplectic lattices. For the sake of completeness we give a proof.
\begin{lemma}\label{lem:bound_integer_points}
	Let $d\in \N$ and $r$ be a positive real number. Then
	\[
	\left|\Z^d\cap B_d(r)\right|\le \sigma_d (r+\sqrt{d}/2)^d.
	\]
\end{lemma}
\begin{proof}
	Around each integer point $z$ in $B_d(r)$ place an open $d$-dimensional cube centered at $z$ with side-length $1$, denoted by $C(z)$. Every cube $C(z)$ has $d$-dimensional volume $1$. The cubes around different integer points do not intersect. Moreover, for all $z$ we have $C(z)\subset B_d(r+\sqrt{d}/2)$. Denote the $d$-dimensional volume of a (measurable) set $S$ by $V_d(S)$. Then from the above we deduce
\begin{align*}
	\left|\Z^d\cap B_d(r)\right|=V_d\left(\bigcup_{z\in \Z^d\cap B_d(r)} C(z)\right)\le & V_d\big(B_d(r+\sqrt{d}/2)\big)^d\\
	       = &\sigma_d\cdot \big(r+\sqrt{d}/2\big)^d.
\end{align*}
\end{proof}
\paragraph{Some useful bounds}
The probabilistic analysis of crucial parameters of random symplectic lattices below also uses the following well known bounds on factorials and volumes of high-dimensional balls.
\begin{lemma}\label{lem:bounds}As before denote by $\sigma_d$  the volume of the unit ball $B_d(1)$ in $\R^d$. Then
\begin{itemize}
	\item[(1)] $\sqrt{2\pi d}\left(\frac {d}{e}\right)^{d} e^{\frac{1}{12d+1}}\le d! \le \sqrt{2\pi d}\left(\frac {d}{e}\right)^{d} e^{\frac{1}{12d}} $
	\item[(2)] $\sqrt{2\pi d}^{1/(2d)}\sqrt{\frac{d}{\pi e}}\le \left(\frac{1}{\sigma_{2d}}\right)^{1/(2d)}\le \sqrt{2\pi d}^{1/(2d)}\sqrt{\frac{d}{\pi e}}e^{\frac{1}{12d^2}}$
\end{itemize}	
\end{lemma}
The first statement of the lemma is an explicit version of Stirling's formula due to Robbins~\cite{robbins1955remark}. The second one follows from the first statement and $\sigma_{2d}=\frac{\pi^d}{d!}$.
The second statement of the lemma shows that up to low order terms $(1/\sigma_{2d})^{1/2d}$ and $\sqrt{d/\pi e}$ coincide. More precisely, 
\begin{equation}
	 \left(\frac{1}{\sigma_{2d}}\right)^{1/(2d)}=\left(1+o\bigl(\frac{\ln(d)}{d}\bigr)\right) \sqrt{\frac{d}{\pi e}}.
\end{equation}

%% file: Constructions.tex
\section{Constructing symplectic matrices from $\sis$ lattices}\label{sec:construction}
Conrad et al. showed how to construct so-called GKP codes (Gotttesman, Kitaev, Preskill) from symplectic lattices $\cL$ with large minimal distance $\lambda_1(\cL)$. We construct such lattices from various types of \emph{modular lattices} as used in (post-quantum) cryptography.
So called modular lattices play an important role in modern lattice-based cryptography. In the most general case, such a lattice is defined via a  matrix $\vb{A}\in \Z^{n\times m}, m>n,$ and an integer $q$. Then one sets
\[
\latperp(\vb{A}):=\{\vb{v}\in \Z^m:\vb{A}\cdot \vb{v}=0\mod q\}.
\]
One easily checks, that this indeed defines a lattice. 
These lattices appear in cryptography in the \emph{short integer solutions problem ($\sis$)}. Mostly for efficiency reasons, nowadays cryptographers prefer \emph{modular ring lattices}. As it turns out, these are also useful in the explicit construction of symplectic lattices and GKP codes. NTRU lattices as used by Conrad et al~\cite{conrad2024good}, in many ways are just special modular ring lattices.  We define these lattices later.

Let $\vb{H}\in \Z^{n\times n}$ be a symmetric matrix, i.e. $\vb{H}=\vb{H}^T$,  and let $q\in \N$. Consider the $2n\times 2n$ matrix
\begin{equation}\label{eqn:q-symplectic}
M_\vb{H}=\left(
\begin{matrix}
 \vb{I}_n & \vb{H}\\
 \vb{0}_n & q\vb{I}_n 
\end{matrix}
\right).
\end{equation}

A simple calculation shows that $M_\vb{H}$ is $q$-symplectic. The rows of $M_\vb{H}$ generate a $2n$-dimensional lattice 
$\cL_\vb{H}=\cL(M_\vb{H})$, i.e. a $q$-symplectic lattice. 

Next, define
\[
A_\vb{H}:=\left(\vb{H}\mid -\vb{I}_n\right)\in \Z^{n\times 2n}
\]
and consider lattice $\latperp(A(\vb{H}))$. It has a basis given by the rows of matrix $M_\vb{H}$
defined above, i.e.
\begin{equation}\label{eq:def_lat_perp_original}
	\latperp(\vb{A}(\vb{H}))=\cL(M_\vb{H}).
\end{equation} 
Simplifying notation, we set
\begin{equation}\label{eq:def_lat_perp}
\latperp(\vb{H}):=\latperp(\vb{A}(\vb{H})).
\end{equation}
As noted above, if $\vb{H}$ is symmetric, then $\latperp(\vb{H})$ is $q$-symplectic and $1/\sqrt{q}\latperp(\vb{H})$ is symplectic. We call lattices of the form $1/\sqrt{q}\latperp(\vb{H}),$ $\vb{H}$ symmetric, \emph{symplectic $\sis$ lattices}.

Symplectic $\sis$ lattices are also used by Conrad et al.\ (see Proposition 3 in \cite{conrad2024good}). However, they did not interpret them as $\sis$ lattices.

%% file: Analysis.tex
\section{Analysis of minimal distance}\label{sec:analysis_distance}
In this section we prove that with probability close to $1$ lattices $\latperp(\vb{H})$ have minimal distance $\approx \sqrt{qn/\pi e}$. Here the probability is with respect to  the uniform distribution over symmetric matrices $\vb{H}$ in $\Z_q^{n\times n}, n\in \N$. 
\begin{definition}
	For $q,n\in\N$, $q$ a prime power,  denote the set of symmetric matrices in $\F_q^{n\times n}$ by $\matsym{q}{n}$. 
\end{definition}
To generate an element of $\matsym{q}{n}$ uniformly at random proceed in two steps
\begin{enumerate}
	\item choose ${n+1 \choose 2}$ elements $h_{ij},1\le i,j\le n, j\ge i.$, uniformly at random from $\F_q$,
	\item for $1\le i\le n, 1\le j<i$, set $h_{ij}=h_{ji}$ 
\end{enumerate}
The analysis of the minimal distance proceeds in two steps and follows standard arguments in lattice-based cryptography (see for example~\cite{micciancio2011geometry}), but taking into account that we work with symmetric matrices.  In the first step of the analysis, we determine for fixed $\vb{z}\in \Z^{2kn},$
\begin{equation}\label{eqn:prob_element_inv_lattice}
	\Pr[\vb{z}\in \latperp(\vb{H}):\vb{H}\leftarrow \matsym{q}{n}].
\end{equation}
Then we use an upper bound on short elements in $\Z^{2n}$  (see Lemma~\ref{lem:bound_integer_points}) and the union bound to obtain the main result.
\subsection{Analyzing the probability}
First, we need a general lemma on solutions to random symmetric linear equations. For $q$ a prime power, denote by $\F_q$ the field with $q$ elements. 
\begin{lemma}\label{lem:solutions_symmetric_matrix}
Let $\vb{z}_i\in \F_q^{2n},i=1,2, \vb{z}=(\vb{z}_1,\vb{z}_2)\ne \vb{0}$. Then 
\[
\Pr[\vb{H}\vb{z}_1=\vb{z}_2; \vb{H}\leftarrow \matsym{q}{n}]=\frac{1}{q^{n}}.
\]		
\end{lemma}
Note that this is the same bound that we obtain by choosing $\vb{H}$ uniformly at random from $\Z_q^{n\times n}$. 
\begin{proof}
	 If $\vb{z}_1=\vb{0}$, then $\vb{z}_2=\vb{0}\mod q,$ which is impossible since $\vb{z}\ne \vb{0}$ by assumption. For any permutation matrix $\vb{P}$ and symmetric $\vb{H}$, matrix $\vb{P}^T\vb{H}\vb{P}$ is symmetric. Therefore, by permuting the rows and columns of $\vb{H}$ and $\vb{z}_1,\vb{z}_2$, we may without loss of generality assume that the last $m$ entries of $\vb{z}_1$ are non-zero and the remaining entries are zero. Next, write a symmetric matrix $\vb{H}\in \F_q^{n\times n}$ as
 \[
 \vb{H}=
\left(
\begin{matrix}
	\vb{H}_{11} & \vb{H}_{12}\\
	\vb{H}_{21} & \vb{H}_{22}
\end{matrix}
\right),
 \]
 where
 \begin{itemize}
 	\item $\vb{H}_{11}\in \Z_q^{n-m\times n-m}, \vb{H}_{22}\in \F_q^{m\times m}$ are symmetric
 	\item $\vb{H}_{21}=\vb{H}_{12}^T\in \F_q^{m\times n-m}$.
 \end{itemize}
 Assume that all elements in $\vb{H}$ are fixed, except the elements on the diagonal of $\vb{H}_{22}$ and the last column of $\vb{H}_{12}$ (which is also the last row of $\vb{H}_{21}$.) Hence, overall $n$ elements are not fixed. Then, there is a unique choice in $\F_q^n$ for the elements not fixed such that $\vb{H}\vb{z}_1=\vb{z}_2\mod q$. This proves the lemma.
\end{proof}
\begin{lemma}\label{lem:prob_element_lattice}
Let $q$ be prime and  $\vb{z}\in \Z^{2n}, \vb{z}\ne \vb{0}\mod q$. Then 
\[
\Pr[\vb{z}\in \latperp(\vb{H}); \vb{H}\leftarrow \matsym{q}{n}]=\frac{1}{q^{n}}.
\]	
\end{lemma}
\begin{proof}
Write $\vb{z}=(\vb{z}_1,\vb{z}_2), \vb{z}_1,\vb{z}_2\in \Z^n$. Then 
\[
\vb{z}\in \latperp(\vb{H})\Leftrightarrow \vb{H}\vb{z}_1=\vb{z}_2 \mod q.
\]
The lemma follows from the previous lemma.
 \end{proof}
 \subsection{Bounding the minimal distance}
\begin{theorem}\label{thm:main_probability_bound}
Let $n,q,r\in \N$, where $q$ is prime and $r\sqrt{q}<q/2$. Then
\begin{equation}\label{eqn:main_prob_minimal_distance}
	\Pr[\lambda_1(\latperp(\vb{H})<r\sqrt{q}; \vb{H}\leftarrow \matsym{q}{n}] \le \sigma_{2n}\big(r+\sqrt{\frac{n}{2q}}\big)^{2n}. 
\end{equation}

\end{theorem}
\begin{proof}
To simplify notation, define event
\[
E:=\{\vb{H}\in \zqn: \lambda_1(\vb{H})<\radius\},
\] 
and denote by $\bar{E}$ its complimentary event. 
Hence, we need to prove that
\[
\Pr[ E(\vb{H});\vb{H}\leftarrow \matsym{q}{n}]\le \sigma_{2n}\big(r+\sqrt{\frac{n}{2q}}\big)^{2n}.
\]
By the union bound
\begin{equation*} 
	\begin{split} 
				\Pr[\vb{H}\in E; \vb{H}\leftarrow \matsym{q}{n}] & = \Pr[\exists \vb{v}\in B_{2n}(\radius)\cap \latperp(\vb{H})\setminus \{\vb{0}\}; \vb{H}\leftarrow \matsym{q}{n}]\\
		 & \le \sum_{\vb{v}\in B_{2n}(\radius)\cap \Z^{2n}\setminus\{\vb{0}\}}\Pr[\vb{v}\in \latperp(\vb{H}); \vb{H}\leftarrow \matsym{q}{n}].
	\end{split}
\end{equation*}
 By Lemma~\ref{lem:bound_integer_points}
\[
|B_{2n}(\radius)\cap \Z^{2n}\setminus \{\vb{0}\}|\le \sigma_{2n}\big(\radius +\frac{\sqrt{2n}}{2}\big)^{2n}=\sigma_{2n}\big(r+\sqrt{\frac{n}{2q}}\big)^{2n} q^n.
\]
Since $\radius<q/2$, elements in $B_{2n}(\radius)\cap \Z^{2n}\setminus \{\vb{0}\}$ are $\ne \vb{0}\mod q$.  Hence, combining the upper bound on $\Pr[\vb{H}\in E; \vb{H}\leftarrow \matsym{q}{n}]$ with 
Lemma~\ref{lem:prob_element_lattice} proves the theorem.
\end{proof}
 \begin{corollary}\label{cor:existence_symplectic_lattice}
Let $n,q\in \N, r\in \R^{> 0}$, where $q$ is prime and $r<\sqrt{q}/2$. Choosing $\vb{H}$ uniformly at random from $\matsym{q}{n}$ and setting
\[
\cL:=\frac{1}{\sqrt{q}}\latperp(\vb{H})
\]
results in a symplectic lattice $\cL$. Except with probability at most
\[
p:= \sigma_{2n}\big(r+\sqrt{\frac{n}{2q}}\big)^{2n}
\]
lattice $\cL$ satisfies $\lambda_1(\cL)\ge r.$
 \end{corollary}
 \begin{proof}
 By Lemma~\ref{lem:q-symplectic_symplectic}, $\cL=\frac{1}{\sqrt{q}}\latperp(\vb{H})$ is symplectic. Its minimal distance $\lambda_1(\cL)$ is given by $\frac{1}{\sqrt{q}}\lambda_1(\latperp(\vb{H})).$ The corollary follows from Theorem~\ref{thm:main_probability_bound} .
 \end{proof}

%% file: MainTheorems.tex
\section{Main results}\label{sec:main_results}
In this section we discuss several parameter choices in Theorem~\ref{thm:main_probability_bound}. Some of these choices lead to a constructive version of a theorem by Buser and Sarnak~\cite{buser1994period} on the existence of symplectic lattices with bounded minimal distance, albeit with slightly worse constants than in~\cite{buser1994period}. Other choices lead to variants of a conjecture by Conrad, Eisert and Seifert~\cite{conrad2024good} on the existence and construction of good GKP codes. In both cases, different parameter choices lead to different combinations of success probability and minimal lattice or code distance. 
\subsection{Symplectic lattices with large minimal distance}
\begin{theorem}[Existence of symplectic lattices with large minimal distance]\label{thm:existence_symplectic}
Choose $\vb{H}$ uniformly at random from $\matsym{q}{n}$ and set
\[
\cL:=\frac{1}{\sqrt{q}}\latperp(\vb{H}).
\]
\begin{enumerate}
	\item For any $\epsilon>0$ and $q\ge \frac{2n\cdot \sigma_{2n}^{1/n}}{\epsilon^2}$, with non-zero probability
	\[
	\lambda_1(\cL)\ge	(1-\epsilon)(1/\sigma_{2n})^{1/2n},
	\]
	where $\sigma_{2n}$ denotes the volume of the unit ball in $\R^{2n}$.
	\item For $q\ge \frac{\pi e n^2}{2}$, with probability at least $1-\frac{e^2}{\sqrt{2\pi n}} =1-\Omega(1/\sqrt{n})$ 
	\[
	\lambda_1(\cL)\ge \sqrt{\frac{n}{\pi e}},
	\]
	\item For $q\ge 2\pi en^{3/2}$, with probability at least $1-e^{-n^{1/4}}$
		\[
		\lambda_1(\cL)\ge 
	\left(1-\frac{1}{n^{3/4}}\right)\sqrt{\frac{n}{\pi e}}.
	\]

\end{enumerate}
\end{theorem}
\begin{proof}
For the first part of the theorem, use Corollary~\ref{cor:existence_symplectic_lattice} with 
\[
	p:= (1-\epsilon/2)^{2n} 
\]
to obtain with probability $1-	p>0$ a symplectic lattice $\cL$ with
\[
\lambda_1(\cL)\ge 	(1-\epsilon)(1/\sigma_{2n})^{1/2n}.
\]
For the second part of the theorem, set
	\[
	p:=\frac{e^2}{\sqrt{2\pi n}}
	\]
	and apply Corollary~\ref{cor:existence_symplectic_lattice}. This shows that with probability at least $1-p$ the lattice $\cL= \latperp(\vb{H})$ has distance 
	\[
	\lambda_1(\cL)\ge (p/\sigma_{2n})^{1/2n}-\sqrt{n/2q}.
	\]
	By the choice of $q$
	\[
	\sqrt{n/2q}\le \frac{1}{n}\sqrt{\frac{n}{\pi e}}.
	\]
	Using Lemma~\ref{lem:bounds} (2) and by our choice of $p$,
	\[
	(p/\sigma_{2n})^{1/2n}\ge e^{1/n} \sqrt{\frac{n}{\pi e}} \ge \left(1+\frac{1}{n}\right)\sqrt{\frac{n}{\pi e}}.
	\]
	The third part of the theorem follows as the previous two parts, by choosing
	\[
	p:=e^{-n^{1/4}}
	\]
	and using
	\[
	(p/\sigma_{2n})^{1/2n}\ge \left(1-\frac{1}{2n^{3/4}}\right)\sqrt{\frac{n}{\pi e}} \quad \text{and $\sqrt{\frac{n}{2q}}\le \frac{1}{2n^{3/4}}\sqrt{\frac{n}{\pi e}}$}.
	\]
\end{proof}
\textit{Discussion}
\begin{itemize}
    \item Buser and Sarnak~\cite{buser1994period} prove the existence of $2n$-dimensional lattices with first successive minimum
\[
\lambda_1(\cL)\ge (2/\sigma_{2n})^{1/2n}\approx (1+1/2n) (1/\sigma_{2n})^{1/2n}.
\]
Since we can take $\epsilon>0$ arbitrary, the bound by Buser and Sarnak is better than the bound in the first part of Theorem~\ref{thm:existence_symplectic} by a factor of $\approx 1+1/2n$. The other two parts show that from symplectic $\sis$ lattices we can get, with high probability, symplectic lattices with distance close to $\sqrt{n/\pi e}$, which differs from the Buser, Sarnak bound only by small factors, approaching $1$ with $n\rightarrow\infty$.

\item The lower bounds on $q$ in Theorem~\ref{thm:existence_symplectic} are required  to guarantee probabilities close to $1$. Ultimately, they can be traced back to the additive terms of $\sqrt{d}/2$ in Lemma~\ref{lem:bound_integer_points}. Many researchers believe that this additional term of $\sqrt{d}/2$ is not necessary, or can be replaced by something much smaller. If we can ignore it, then the lower bounds on $q$ can be reduced to roughly $n$ (going back to the requirement $r\sqrt{q}/2 <q$ in Theorem~\ref{thm:main_probability_bound}). The experiments and simulations that we present in Section~\ref{sec:simulation} confirm that we can choose $q$ much smaller than indicated by Theorem~\ref{thm:existence_symplectic} and still obtain symplectic lattices and GKP with large minimal distance.
\end{itemize}
\subsection{GKP codes with large distance}
GKP codes~\cite{gottesman2001encoding}, also see~\cite{conrad2022gottesman, conrad2024good}, are \emph{stabilizer codes} acting on the Hilbert space of $n$ bosonic modes. The \emph{stabilizer group} of a GKP code is given by a group of so-called \emph{displacement operators}. For a GKP code acting on $n$ modes, this group is isomorphic to a $2n$-dimensional  \emph{symplectically integral} lattice, i.e.\ a lattice $\cL$ with 
\[
\forall \vb{u},\vb{v}\in \cL: \vb{u} \vb{J}_{2n}\vb{v}^T\in \Z.
\]
Any symplectic lattice $\cL$  is symplectically integral and, therefore defines a GKP code.
The distance of a GKP code $\cC$ with stabilizer group isomorphic to symplectically integral lattice $\cL$ is given by 
\[
\Delta(\cC):=\min \{\|\vb{u}\| \mid \vb{u}\in \latcross\setminus \cL\},
\]
where as before $\latcross$ denotes the symplectic dual of $\cL$ (see for example~\cite{conrad2022gottesman}). Hence
\[
\Delta(\cC)\ge \lambda_1(\latcross).
\]
If lattice $\cL$ is symplectic and $\lambda \in \N$, then $(\sqrt{\lambda} \cL)^\times=\frac{1}{\sqrt{\lambda}}\cL$. Hence for GKP code $\cC$ with stabilizer group isomorphic to $\sqrt{\lambda}\cL$, we obtain
\begin{equation}\label{eq:code_distance_equality}
\Delta(\cC)=\lambda_1((\sqrt{\lambda}\cL)^\times)=\frac{1}{\sqrt{\lambda}}\lambda_1(\cL).
\end{equation} 
(see~\cite{conrad2022gottesman} and Lemma~\ref{lem:symplectic_self_dual}). 
Equation~(\ref{eq:code_distance_equality}) together with Theorem~\ref{thm:existence_symplectic}, part 2 and 3, yields the following theorem.
\begin{theorem}[Good GKP codes from $\sis$ lattices]\label{thm:construction_good_GKP_codes}
A GKP code $\cC$ with stabilizer group $\cL=\sqrt{\lambda/q}\latperp(\vb{H})$, where $\latperp(\vb{H})$ is as in (\ref{eq:def_lat_perp}) has $\lambda^n$ logical dimensions. Over the choice of $\vb{H}\leftarrow \matsym{q}{n}$,  the distance $\Delta(\cC)$ of $\cC$  satisfies the following.
\begin{enumerate}
	\item For $q\ge \frac{\pi e n^2}{2}$, with probability at least $1-\frac{e^2}{\sqrt{2\pi n}} =1-\Omega(1/\sqrt{n})$ 
	\[
	\Delta(\cC)\ge \sqrt{\frac{n}{\lambda\pi e}},
	\]
	\item For $q\ge 2\pi en^{3/2}$, with probability at least $1-e^{-n^{1/4}}$
		\[
		\Delta(\cC)\ge 
	\left(1-\frac{1}{n^{3/4}}\right)\sqrt{\frac{n}{\lambda\pi e}}.
	\]
\end{enumerate}
\end{theorem}
\textit{Discussion} Conrad et al.\ ~\cite{conrad2024good} conjecture that certain GKP code constructions based on NTRU lattices from cryptography yield (with high probability) GKP codes with distance $\sqrt{n/\lambda\pi e}$. Part one of the above theorem shows that GKP codes with distance $\ge  \sqrt{n/\lambda \pi e}$ can be constructed from $\sis$ lattices. Part two shows that, by going slightly below $\sqrt{n/\lambda \pi e}$, we can achieve distance roughly $\sqrt{n/\lambda \pi e}$ from $\sis$ lattices with probability exponentially close to $1$. 

\subsection{Upper bounds for minimal distance and covering radius and }
In this section, we prove some upper bounds on the minimal distance of symplectic $\sis$ lattices and GKP codes obtained from such lattices, as well as on the covering radius of symplectic $\sis$ lattices.  All bounds are either well-known or can be proven easily using techniques known from general $\sis$ lattices. For the sake of completeness we include proofs.
\begin{lemma}\label{lem:upper_bound_distance_symplectic}
Let $\cL$ be a $2n$-dimensional symplectic lattice and $\lambda \in \N$. Then 
\begin{enumerate}
	\item $\lambda_1(\cL)\le \sqrt{2n}$
	\item $\lambda_1((\sqrt{\lambda} \cL)^\times) \le \sqrt{\frac{2n}{\lambda}} $.
\end{enumerate}
\end{lemma}
\begin{proof}
	Since $\cL$ is symplectic, $\det(\cL)=1$. By Minkowski's first theorem,
	\[
	\lambda_1(\cL)\le \sqrt{2n}\det(\cL)^{1/2n}=\sqrt{2n}.
	\]
	This proves the first part of the lemma. The second follows from part 1 and Lemma~\ref{lem:symplectic_self_dual}.
\end{proof}
The next corollary follows from the previous lemma and~(\ref{eq:code_distance_equality}).
\begin{corollary}\label{cor:upper_bound_GKP_code}
	Let $\cL$ be a $2n$-dimensional symplectic lattice, $\lambda\in \N$ and $\cC$ the GKP code with stabilizer group isomorphic to $\sqrt{\lambda}\cL$. Then 
	\[
	\Delta(\cC)\le \sqrt{\frac{2n}{\lambda}} .
	\]
\end{corollary}

\begin{theorem}\label{thm:covering_radius_symplectic_general}
Let $\lambda\in \N$. Choose $\vb{H}$ uniformly at random from $\matsym{q}{n}$ and set
\[
\cL:=\frac{1}{\sqrt{q}}\latperp(\vb{H}).
\]
\begin{enumerate}
	\item For $q\ge \frac{\pi e n^2}{2}$, with probability at least $1-\frac{e^2}{\sqrt{2\pi n}} =1-\Omega(1/\sqrt{n})$ 
	\[
	\rho((\sqrt{\lambda}\cL)^\times)\le \frac{1}{2}\sqrt{\frac{\pi e n}{\lambda}} .
	\]
	\item For $q\ge 2\pi en^{3/2}$, with probability at least $1-e^{-n^{1/4}}$
	\[
	\rho((\sqrt{\lambda}\cL)^\times)\le 
	\left(1-\frac{1}{n^{3/4}}\right)^{-1}\frac{1}{2}\sqrt{\frac{\pi e n}{\lambda}}.
	\]
\end{enumerate}
\end{theorem}
\begin{proof}
We only prove the first part, the second part can be shown analogously. Lattice $\cL$ is symplectic and by Theorem~\ref{thm:symplectic_covering_radius} we have 
\[
\rho((\sqrt{\lambda}\cL)^\times)=\frac{1}{\sqrt{\lambda}}\rho(\cL^\times)\le \frac{n}{2\sqrt{\lambda}\lambda_1(\cL)} .
\]
Hence the first part of the theorem follows from part two of Theorem~\ref{thm:existence_symplectic}.
\end{proof}
Using $\lambda_1(\cL)\le 2\rho(\cL)$, true for any lattice $\cL$, using $\Delta(\cC)=\lambda_1((\sqrt{\lambda}\cL)^\times)$, true for GKP code with stabilizer group $\sqrt{\lambda}\cL$, and using the previous theorem, one obtains the upper bound (with high probability) $\Delta(\cC)\le \sqrt{\frac{\pi e n}{\lambda}}$. However, the bound in Corollary~\ref{cor:upper_bound_GKP_code} is better by a constant.

Together, Theorem~\ref{thm:existence_symplectic}, Lemma~\ref{lem:upper_bound_distance_symplectic} and Theorem~\ref{thm:covering_radius_symplectic_general}, show that for a random symplectic $\sis$ lattice important lattice parameters like the minimal distance and the covering radius differ from each other and (the general upper bound for $\lambda_1(\cL)$ of)  $\sqrt{2n}$ by constants (and with high probability). Theorem~\ref{thm:construction_good_GKP_codes} and Corollary~\ref{cor:upper_bound_GKP_code}, on the other hand, show that the distance of a GKP constructed from random symplectic $\sis$ lattices differs from the general upper bound $\sqrt{\frac{2n}{\lambda}}$ only by constant factors (with high probability).

%% file: Bounded_Distance_Decoding.tex
\section{Bounded distance decoding}\label{sec:bdd}
As described by Conrad et al.\ ~\cite{conrad2022gottesman}, one way to heuristically solve the maximum likelihood decoding problem for GKP codes requires a solution to the \emph{bounded distance decoding ($\bdd$) problem} for the symplectic dual of the lattices defining a GKP code.  Therefore, in this section we describe a simple algorithm for the $\bdd$ problem for the symplectic dual of the  lattices used in Theorem~\ref{thm:construction_good_GKP_codes}. 
Recall that for a lattice $\cL$ and a point $\vb{t}\in \spann(\cL)$ we denote by $\dist(\vb{t},\cL)$ the distance of $\vb{t}$ to $\cL$
\[
\dist(\vb{t},\cL):=\min \{\|\vb{t}-\vb{u} \| \mid \vb{u}\in \cL\}.
\]
\begin{definition}[Problem $\bdd_{R}$]\label{def:bdd} Given a basis $\vb{B}$ for lattice $\cL=\cL(\vb{B})$, a distance bound $R\in \R^{\ge 0}$ and a point $\vb{t}\in \spann(\cL)$, such that  $\dist(\vb{t},\cL)\le R$, find a lattice point $\vb{v}\in \cL$ with $\|\vb{t}-\vb{v}\|\le R$. We denote this by $\bdd(\cL,\vb{t},R)$. 
\end{definition}
Often, $R$ is expressed as a multiple of $\lambda_1(\cL)$ or as a multiple of its covering radius $\rho(\cL)$. Note, that if $R<\lambda_1(\cL)/2$ then a solution to $\bdd(\cL,\vb{t},R)$ is unique. 

To solve the BDD problem for lattices $(\sqrt{\lambda}\cL)^\times=\frac{1}{\sqrt{\lambda}}\cL,$ where $\cL=\frac{1}{\sqrt{q}}\latperp(\vb{H})$ is symplectic,  we propose the following simple algorithm $\bdd$. The input to the algorithm is a symmetric matrix $\vb{H}\in \Z^{n\times n}$, a prime $q$, an integer $\lambda$, together defining the lattice $\frac{1}{\sqrt{\lambda q}}\latperp(\vb{H})$, and a target vector $v=(v_1,\ldots,v_n,v_{n+1},\ldots, v_{2n})$. It outputs an integer vector $c=(c_1,\ldots,c_n,c_{n+1},\ldots,c_{2n})$. Implicitly, the algorithm uses the basis $\frac{1}{\sqrt{\lambda q}}M(\vb{H})$ for lattice $\frac{1}{\sqrt{\lambda q}}\latperp(\vb{H})$ (see~(\ref{eq:def_lat_perp_original}) and~(\ref{eq:def_lat_perp})).
\begin{center}

\fbox{
		\begin{tabular}{ll}
			\multicolumn{2}{l}{$\bddtriv$ on input $\vb{H}\in \Z^{n\times n}, q,\lambda \in \N$ and $\vb{v}\in \Z^{2n}$:}\\
			1. & For  $i=1,\ldots,n$, set $c_i:=\lfloor \sqrt{\lambda q} v_i \rceil$. \\
			2. & Set $\vb{c}':=(c_1,\ldots, c_n), \vb{w}=\vb{v}'-\vb{c}'\vb{H}$, where $\vb{v}':=(v_{n+1},\ldots,v_{2n})$, \\
			  & set $\vb{w}=(w_1,\ldots,w_n)$.\\
			3. & For $i=n+1,\ldots, 2n$, set $c_i:=\lfloor \sqrt{\lambda/ q} w_{i-n} \rceil$. \\
			4. & Output $\vb{c}=(c_1,\ldots,c_{2n})$.
		\end{tabular}
	}
\end{center}
The next theorem follows directly from the algorithm.
\begin{theorem}\label{thm:analysis_algorithm_bdd}
Let $\vb{H}, q, \lambda,$ and $\vb{v}$ be as above. Algorithm $\bddtriv$ computes the closest vector in $\frac{1}{\sqrt{\lambda q}}\latperp(\vb{H})$ to $\vb{v}$ as an integer linear combination of the rows in $M_{\vb{H}}$, if the distance $\dist(\vb{v},\latperp)$ of $\vb{v}$ to $\latperp$ is less than $\frac{1}{2\sqrt{\lambda q}}$.	
\end{theorem}
If we compare the bound $\frac{1}{2\sqrt{\lambda q}}$ from Theorem~\ref{thm:analysis_algorithm_bdd} with 
\begin{itemize}
	\item the lower bounds on the minimal distance and covering radius of lattice $\cL=\frac{1}{\sqrt{\lambda q}}\latperp(\vb{H})$ from Theorem~\ref{thm:existence_symplectic} (part three)  and from Theorem~\ref{thm:covering_radius_symplectic_general} (part three), respectively

	\item the lower bound on the minimal distance of GKP code $\cC$ with stabilizer $\sqrt{\frac{\lambda}{q}}\latperp(\vb{H})$ from Theorem~\ref{thm:construction_good_GKP_codes} (part two)
\end{itemize}
 we see that algorithm $\bddtriv$
\begin{enumerate}
	\item computes the closest lattice vector to target vector $\vb{t}$ if 
	\begin{itemize}
		\item $\dist(\vb{t},\cL)\le \gamma\lambda_1(\cL)$ for $\gamma\le \frac{(1-1/n^{3/4})^{-1}}{2\sqrt{2}n^{5/4}}$,
		\item $\dist(\vb{t},\cL)\le \gamma\rho(\cL)$ for $\gamma\le \frac{(1-1/n^{3/4})^{-1}}{\sqrt{2}\pi e n^{5/4}}$, 
	\end{itemize}
	\item computes the closest code word to target vector $\vb{t}$ if $\dist(\vb{t},\cC)\le \gamma\Delta(\cC)$ for $\gamma\le \frac{(1-1/n^{3/4})^{-1}}{2\sqrt{2}n^{5/4}}$,
\end{enumerate}
 in all cases with probability exponentially close to $1$ over the random choice of $\vb{H}$.

%% file: Simulation.tex
\section{Simulation of main results}\label{sec:simulation}
In this section we show some simulation results corresponding to the main results in the above sections. The simulation results shows the correctness of the main results from another perspective.
\subsection{Environment setting}
Codes for simulation are written with Python and SageMath, and are run on a laptop in durable time. For calculating the distance $\Delta(C)$ for GKP code $C$ or the first successive minimum $\lambda_1(\cL)$ for lattice $\cL$, we use approximation algorithms from the python package fpylll\cite{fpylll}.
\subsection{Distance of GKP code from SIS lattices}
In this section we show the simulation results of Theorem \ref{thm:construction_good_GKP_codes}. In the simulation, $\lambda=2$ is fixed as coding into qubits is mostly used. Several parameter combinations of $(q,n)$ are applied. In each $(q,n)$ combination, we sample 1000 $\vb{H}\leftarrow \matsym{q}{n}$, then construct the lattice $\cL$ of GKP code $\cC$ with $\vb{H}$ via the method in the theorem, and finally calculate the first successive minimum $\lambda_1(\cL)$ for lattice $\cL$, i.e. the code distance $\Delta(C)$ for GKP code $C$ by applying BKZ reduction \cite{schnorr1987hierarchy}. For simplicity, we only show parts of simulation results. Other plots and complete data can be found in the supplementary material.

\begin{figure}[htbp]
  \centering

  \begin{subfigure}[b]{0.49\textwidth}
    \centering
    \includegraphics[width=\linewidth]{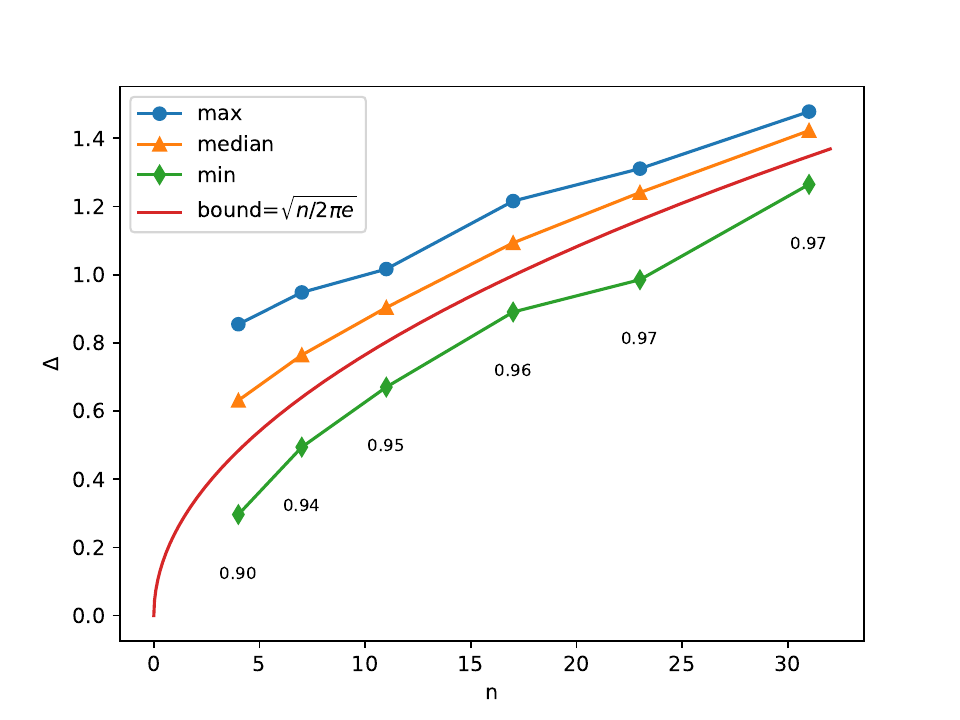}
    \caption{fix $q=256$}
    \label{code_distance_sim_a}
  \end{subfigure}
  \hfill
  \begin{subfigure}[b]{0.49\textwidth}
    \centering
    \includegraphics[width=\linewidth]{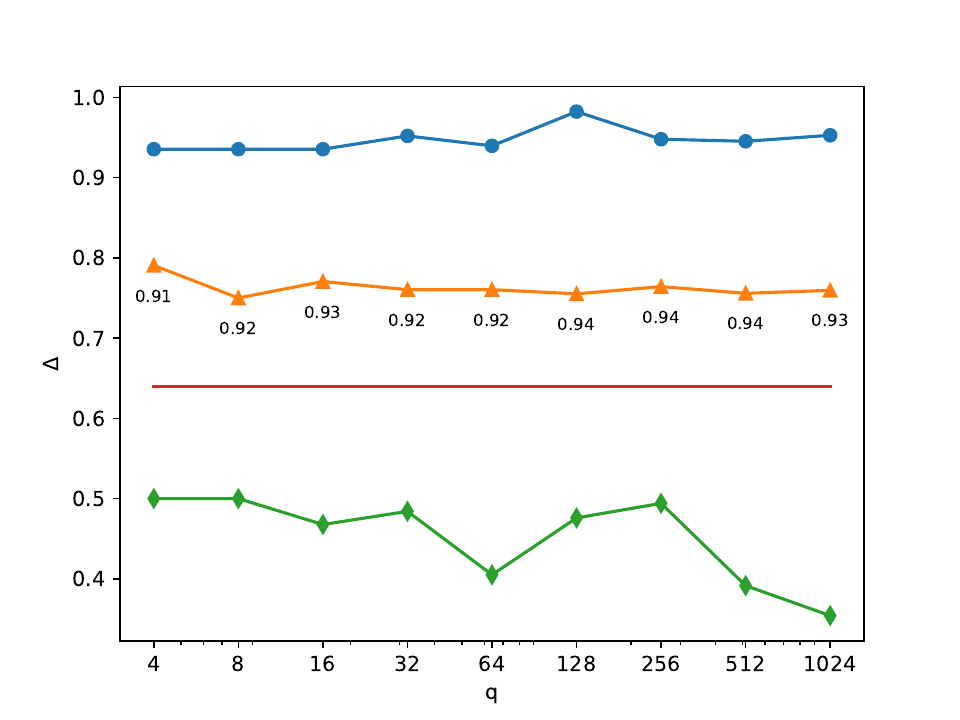}
    \caption{fix $n=7$}
    \label{code_distance_sim_b}
  \end{subfigure}

  \caption{Code distance of GKP code generated from SIS lattice, (a) with $q=256$ fixed, vary $n$ as horizontal axis; (b) with $n=7$ fixed, vary $q$ as horizontal axis.}
  \label{fig:code_distance_sim}
\end{figure}

Fig. \ref{fig:code_distance_sim} shows two cases of code distance varying with $n$ and $q$. In the two graphs, the maximum, median and minimum of code distance of 1000 GKP code samples are plotted. The numbers below each data point in the plot show the rate of samples with code distance $\Delta>\sqrt{n/2\pi e}$. Either in Fig. \ref{fig:code_distance_sim} (a) or Fig. \ref{fig:code_distance_sim} (b), more than 90 percent samples have code distance above the bound $\sqrt{n/2\pi e}$, which correspond to Theorem \ref{thm:construction_good_GKP_codes}. It is worth nothing that this high above bound rate is for all combinations of $(n,q)$ in the simulation, which gives the conjecture that with the construction given by Theorem \ref{thm:construction_good_GKP_codes}, code distance have high probability (close to 1) to be more than the bound $\sqrt{n/\lambda\pi e}$ for all $(n,q)$ combination, i.e. also for small $q\leq\pi en^2/2$ or $q\leq2\pi en^{3/2}$. One can also see from Fig. \ref{fig:code_distance_sim} that Corollary \ref{cor:upper_bound_GKP_code} and Theorem \ref{thm:covering_radius_symplectic_general} are true in the simulation.

\subsection{Trivial decoder for BDD}
In this section we show the decoding performance of the trivial decoder $\bddtriv$, as well as comparing it with NTRU decoder mentioned in \cite{conrad2022gottesman} and Babai decoder using Babai's algorithm \cite{babai1986lovasz}.

The simulation runs as follows: for every combination of $(n,q)$, we first sample 100 GKP codes and pick the one with the largest code distance for decoding. In the decoding stage, we use the Gaussian error model for GKP codes and vary the physical standard deviation $\sigma$. For certain $\sigma$, we generate $10^4$ error samples with Gaussian error model and get the syndromes of these error samples. According to \cite{conrad2022gottesman}, with the knowledge of error syndrome, we can transform the decoding problem into BDD. Then applying $\bddtriv$ (or NTRU decoder or Babai decoder for comparison), we get the residual vector after correction. Finally check if the residual vector can be stabilized and if it commutes with all logical operators; in the lattice perspective, it means whether the residual vector is in the lattice $\cL$ of the GKP code. If the residual vector is in $\cL$, we denote the decoding (correction) \textit{success}, else we denote it \textit{fail}. By calculating the number of failed decoding cases divided by the total decoding cases, we get the decoding error rate, which is used as a flag for evaluating decoding performance. One note is that when comparing decoding performance for trivial, NTRU and Babai decoder, GKP codes are generated via NTRU lattice as mentioned in \cite{conrad2022gottesman}, since the decoding of NTRU decoder relies on the structure of NTRU lattice. For studying the decoding performance of trivial decoder, only GKP codes from SIS lattices are needed. 

\begin{figure}[htbp]
  \centering

  \begin{subfigure}[b]{0.49\textwidth}
    \centering
    \includegraphics[width=\linewidth]{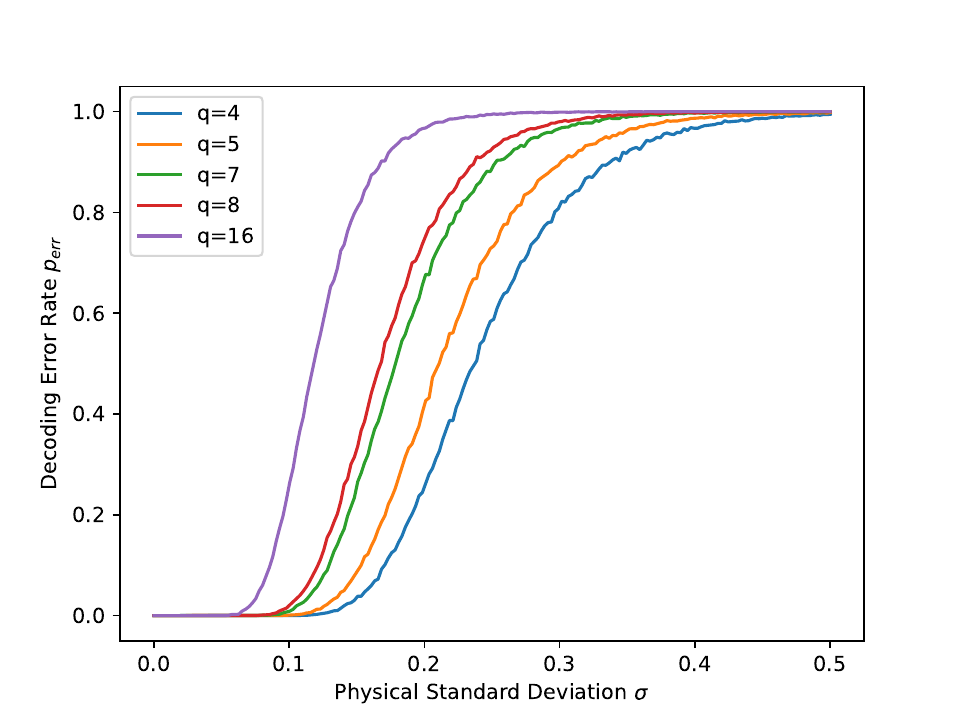}
    \caption{trivial decoder, $n=11$}
    \label{decoding_sim_a}
  \end{subfigure}
  \hfill
  \begin{subfigure}[b]{0.49\textwidth}
    \centering
    \includegraphics[width=\linewidth]{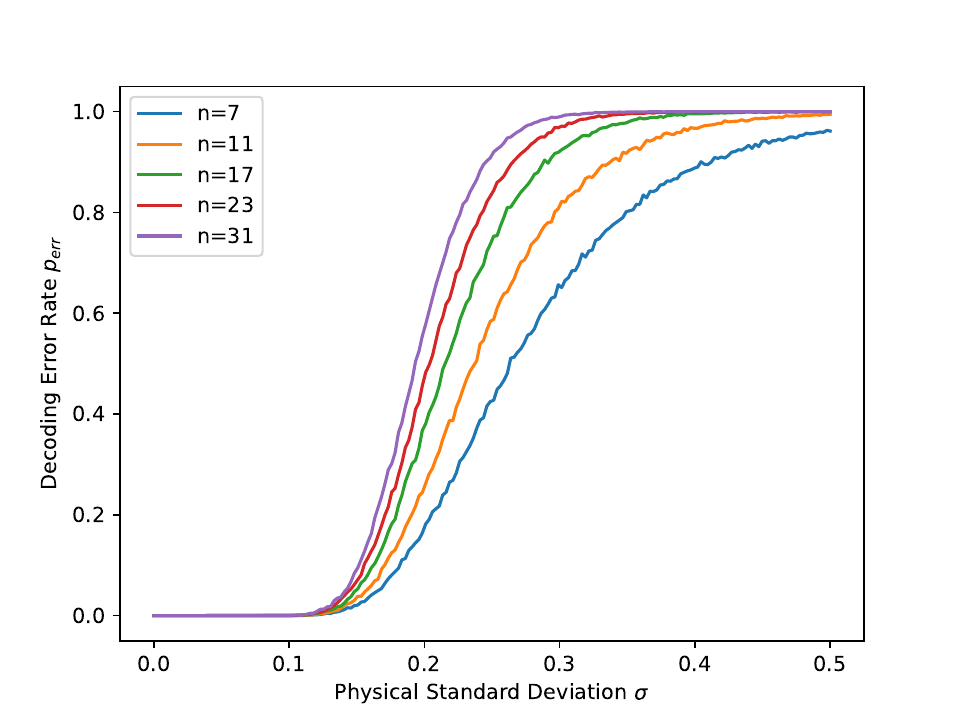}
    \caption{trivial decoder, $q=4$}
    \label{decoding_sim_b}
  \end{subfigure}

  \vspace{0.3cm}

  \begin{subfigure}[b]{0.49\textwidth}
    \centering
    \includegraphics[width=\linewidth]{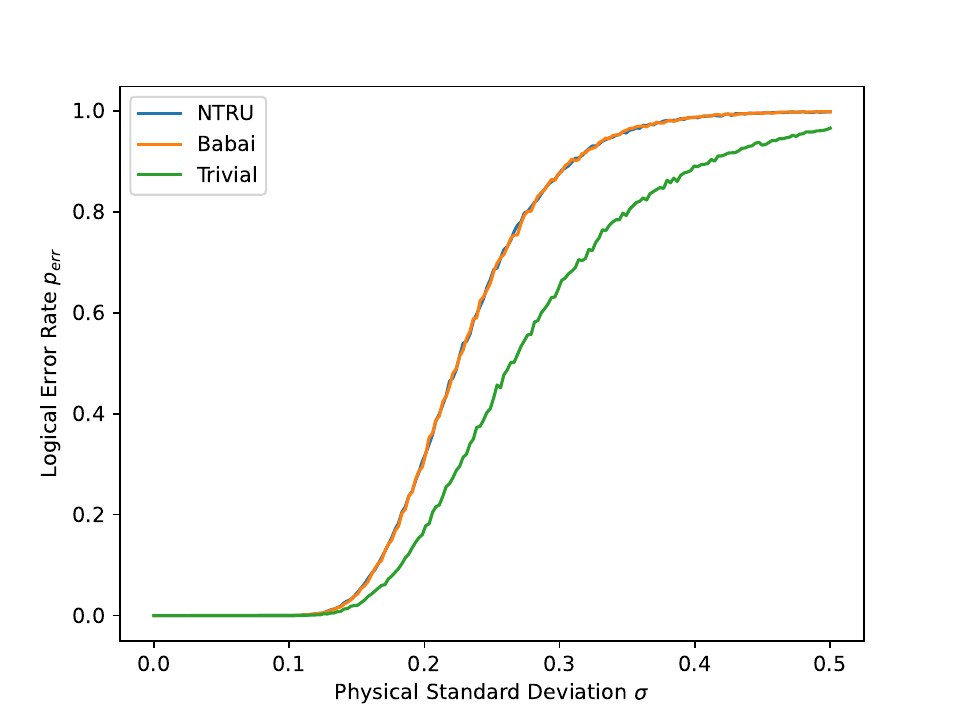}
    \caption{decoder comparison, $n=7$, $q=4$}
    \label{decoding_compare_a}
  \end{subfigure}
  \hfill
  \begin{subfigure}[b]{0.49\textwidth}
    \centering
    \includegraphics[width=\linewidth]{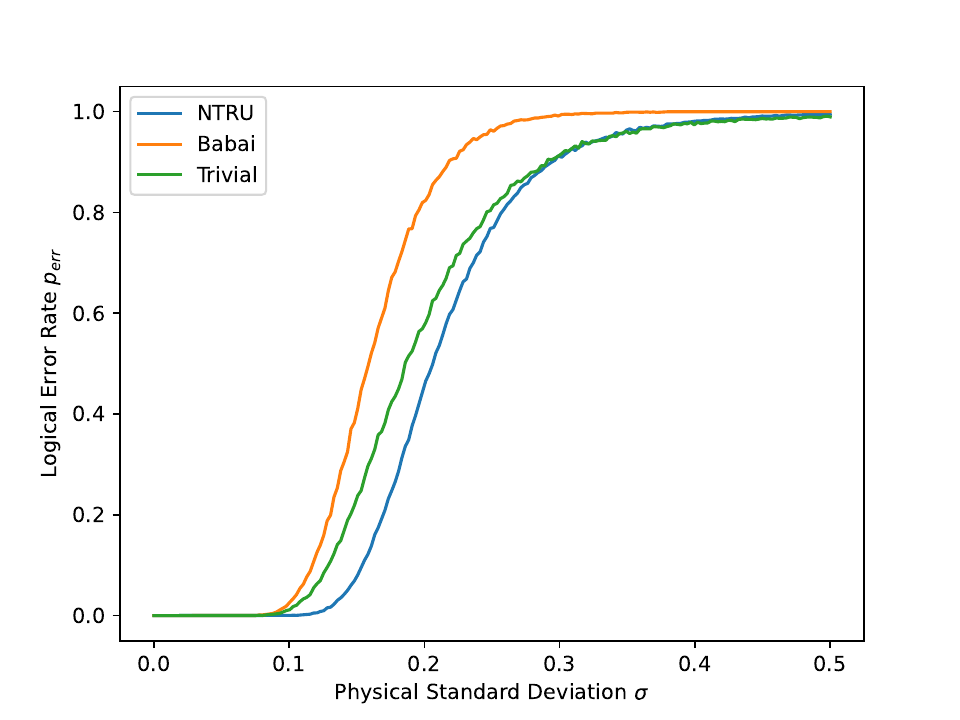}
    \caption{decoder comparison, $n=7$, $q=8$}
    \label{decoding_compare_b}
  \end{subfigure}

  \caption{Decoding under Gaussian displacement noise. (a) Trivial decoder at $n=11$ for $q\in\{4,5,7,8,16\}$. (b) Trivial decoder at $q=4$ for $n\in\{7,11,17,23,31\}$. (c–d) NTRU decoder, Babai decoder, and trivial decoder ($\bddtriv$) at $n=7$ with $q=4$ and $q=8$. Curves show error rate $p_{\mathrm{err}}$ versus physical standard deviation $\sigma$. For each $(n,q)$, we select the code with the largest distance from 100 random candidates and estimate $p_{\mathrm{err}}$ using $10^{4}$ error samples per $\sigma$. For (c–d), codes are generated from NTRU lattices to enable the NTRU decoder.}
  \label{fig:decoding_sim}
\end{figure}

Fig. \ref{fig:decoding_sim} (a) and (b) shows the decoding performance of $\bddtriv$. From Fig. \ref{fig:decoding_sim} (a) we can see that the decoding performance for trivial decoder gets worse as $q$ increase, which corresponds with Theorem \ref{thm:analysis_algorithm_bdd}, since the trivial decoder only computes the closest vector correctly when $\dist(\vb{v},\latperp)<\frac{1}{2\sqrt{\lambda q}}$. Notice that $\dist(\vb{v},\latperp)$ represents the amplitude of the real error. Meanwhile in our simulation, errors are generated by Gaussian error model. Therefore as $q$ increase, $\frac{1}{2\sqrt{\lambda q}}$ becomes smaller. There's less probability to have $\dist(\vb{v},\latperp)<\frac{1}{2\sqrt{\lambda q}}$ for an random Gaussian error, thus the decoding performance drops.

From Fig. \ref{fig:decoding_sim} (b), we can see that the decoding performance for the trivial decoder gets worse as $n$ increases. This also corresponds with Theorem \ref{thm:analysis_algorithm_bdd}, since $\gamma$ decreases as $n$ increases. Therefore, the decoding performance for the trivial decoder gets worse for similar reasons as mentioned above.

Fig. \ref{fig:decoding_sim} (c) and (d) show two cases of the decoding performance of $\bddtriv$, NTRU decoder and Babai decoder. Based on the simulation results, the trivial decoder performs better than both NTRU and Babai decoder when $q$ is small ($q\leq 4$), and performs worse than NTRU but still better than Babai decoder when $q$ is large. Considering that there's no need for picking $q$ large for GKP codes (Since the distance bound $\sqrt{n/2\pi e}$ is irrelevant with $q$) and that the trivial decoder requires less computing resources, the trivial decoder can be considered as a good option for decoding GKP codes generated from lattices.

%% file: Ring_Constructions.tex
\section{Symplectic lattices from $\rsis$ and $\msis$ lattices}\label{sec:ring_symplectic_construction}
In this section we construct symplectic lattices from so-called ring and module  $\sis$ lattices, $\rsis$ and $\msis$ lattices, respectively. Both are special cases of modular lattices. Furthermore, $\rsis$ lattices are a subclass of $\msis$ lattices. We consider $\rsis$ lattices separately to increase readability.
\subsection{Basic definitions and results}
To define ring and module lattices, fix some integer polynomial $\Phi(X)$  and an integer  $q$. These define rings
\[
R:=\Z[X]/(\Phi(X)) \quad \text{and} \quad R_q:=\Z_q[X]/(\Phi(X)).
\] 
Let $\vb{A}\in R_q^{k\times m}$, for  integers $k,m$. Similar, to the definition for modular lattices, set
\[
\latperp(\vb{A}):=\{\vb{v}\in R^m:\vb{A}\cdot \vb{v}=0\mod qR\}.
\]
$\latperp$ is a lattice in the following sense. Set $n:=\deg(\Phi(X))$. Then the elements in $R$ and $R_q$ have representatives of degree at most $n-1$. Hence, we can identify elements in $R$ with their coefficient vectors in $\Z^n$, and elements in $R_q$ can be identified with their coefficient vectors in $\Z_q^n$. More generally, we can identify elements in $R^m$ and $R_q^m$ with elements in $\Z^{nm}$ and $\Z_q^{nm}$, respectively. It follows that for $\vb{A}\in R_q^{k\times m}$ the above definition of $\latperp(\vb{A})$ yields a lattice in $\Z^{nm}$, which we also denote by $\latperp(\vb{A})$. For $k=1$ these lattices appear in the ring version of the short integer solutions problem (in short $\rsis$). 
For general $k$, lattices $\latperp(\vb{A})$ correspond to instances of the module version of the short integer solution problem (in short $\msis$). In both cases, in $R^m$ the sets $\latperp(\vb{A})$ have an additional module structure. Therefore the lattices $\latperp(\vb{A})$ correspond to embeddings (via the coefficient vectors) of these modules into $\R^{nm}$  and are called \emph{module lattices}\footnote{In algebraic number theory as well as in some parts of lattice based cryptography one uses the  (canonical) embedding via field homomorphisms. But this is not relevant for this paper.}. The special case $k=m=1$ corresponds to so called \emph{ideal lattices}. But these are not relevant for this paper. 

Using matrix representations of multiplications by ring elements, one can see that module lattices as just described are special cases of modular lattices. These matrix representations are crucial for the construction of symplectic lattices and good GKP codes.
\paragraph{Matrix representations} We consider $\ring$ as a finite dimensional $\Q$ vector space. Throughout this paper we fix the basis $1,X,\ldots, X^{n-1}$. Let $h(X)=\sum_{i=0}^{n-1}h_iX^i$ be an element in $R:=\Z[X]/(\Phi(X))$. Consider the linear map
\begin{equation}\label{eqn:function_mult}
\begin{array}{ll}
	\rho(\vb{h}):  R & \longrightarrow R\\
	  b & \longmapsto b \cdot h.
\end{array}
\end{equation}
As a linear map it has a matrix representation with respect to the basis $1,\ldots,X^{n-1}$, which we also denote by $\rho(\vb{h})$. 

The most important rings considered in this paper are rings $R,R_q$ for the polynomials $\phi(X)=X^n\pm 1$. Let $h=h(X)=\sum_{i=0}^{n-1}h_iX^i$ be an element of $R$. For $\Phi(X)=X^m-1$ the matrix representation of $\rho(\vb{h})$ is given by
\[
\rho(\vb{h})=\left(
\begin{matrix}
	h_0 & h_1 & h_2 & \ldots & h_{n-1}\\
	h_{n-1} & h_0 & h_1 & \ldots & h_{n-2}\\
	\vdots &\vdots &\vdots & \ddots & \vdots \\
	h_2 & h_3 & h_4 & \ldots & h_1\\
	h_1 & h_2 & h_3 & \ldots & h_0
\end{matrix}
\right).
\]
For $\Phi(X)=X^n+1$ the matrix representation is given by
\begin{eqnarray}\label{eqn:matrix_representation}
\rho(\vb{h})=\left(
\begin{matrix}
	h_0 & h_1 & h_2 & \ldots & h_{n-1}\\
	-h_{n-1} & h_0 & h_1 & \ldots & h_{n-2}\\
	\vdots &\vdots &\vdots & \ddots & \vdots \\
	-h_2 & -h_3 & -h_4 & \ldots & h_1\\
	-h_1 & -h_2 & -h_3 & \ldots & h_0
\end{matrix}
\right).
\end{eqnarray}
\paragraph{Factorization of cyclotomic polynomials} In this paper we consider  rings $R=\rring$ and $\rringq$, where $n\in \N$ and $q$ a prime with $\gcd(n,q)=1$.The structure of rings $R=\rring$ and $R_q=\rringq$ is determined by the factorization of the polynomial $X^n+1$ over $\Q$ and modulo $q$. More precisely, over $\Q$
\begin{eqnarray}\label{eqn:factors_xn+1}
X^n+1=\prod_{\substack{m\mid 2n \\ m\nmid n}}Q_m(X),
\end{eqnarray}
where, following~\cite{lidl1994introduction},  $Q_m$ denotes the $m$-th cyclotomic polynomial over  $\Q$. That is,
\[
Q_m(X):=\prod_{\substack{s=1\\\gcd(s,m)=1}}^d (X-\zeta^s),
\]
where $\zeta$ denotes a primitive $m$-th root of unity over $\Q$.  

For $m\in \N$, set $\setn{m}:=\{0,\ldots,m\}$.  Modulo $q$ the polynomials $Q_m$ are no longer irreducible, instead there exist integers $n_m,d_m$ such that $Q_m$ factors into $n_m$ distinct irreducible factors $p_{m,j}$ each if degree $d_m$, where $d_m$ is the order of $q$ modulo $m$. (see e.g. Theorem 2.47 in~\cite{lidl1994introduction}). Consequently, over the field $\F_q$, polynomial $X^n+1$ factors as
\begin{equation}\label{eq:factorizatzion_xn}
X^n+1=\prod_{\substack{m\mid 2n \\ m\nmid n}}\prod_{j_m\in\setn{n_m}} p_{m,j_m}(X),
\end{equation}
where the polynomials $p_{m,j_m}$  are pairwise distinct and irreducible. Hence 
\begin{equation}\label{eqn:ring_isomorphism}
R_q\cong \bigotimes_{\substack{m\mid 2n \\ m\nmid n}}\bigotimes_{j_m|in\setn{n_m}} \Fq{q}(X)/(p_{m,j_m})(X))\cong \bigotimes_{\substack{m\mid 2n \\ m\nmid n}} \left(\Fq{q^{d_m}}\right)^{n_m}.
\end{equation}
\subsection{Symplectic lattices from $\rsis$ lattices}
From now on let $n,q\in\N$ and set $\Phi(X)=X^n+1$. We define rings $R,R_q$ as above. Let $h(X)=\sum_{i=0}^{n-1}h_iX^i$ be an element in $R:=\Z[X]/(\Phi(X))$. Consider the linear map $\rho(\vb{h})$ and its matrix representation with respect to basis
$1,X,\ldots, X^{n-1}$ as defined above (see (\ref{eqn:function_mult}) and \ref{eqn:matrix_representation})). Matrix $\rho(\vb{h})$ is not symmetric. Consequently, matrix $M_{\rho(\vb{h})}$ as in (\ref{eqn:q-symplectic}) is not $q$-symplectic. However, define the following $(n-1)$-dimensional and $n$-dimensional  matrices
\[
\bar{I}_{n-1}:=\left(
\begin{matrix}
	0 & 0& \ldots &0 & 1\\
	0 & 0& \ldots &1 &0\\
	\vdots & \vdots & \ddots & \vdots & \vdots \\
	0 & 1 & \ldots & 0 & 0\\
	1 & 0 & \vdots & 0 & 0
\end{matrix}
\right), 
\quad 
\sigma_n:=\left(
\begin{matrix}
	1& 0\\
	0_n & -\bar{I}_{n-1}
\end{matrix}
\right).
\]
Then the matrix 
\[
\bar{\rho}(\vb{h}):=\sigma_n \cdot \rho(\vb{h}):=
\left(
\begin{matrix}
	h_0 & h_1 & h_2 & \ldots & h_{n-1}\\
	h_1& h_2 & h_3 & \ldots & -h_0\\
	\vdots & \vdots &\vdots & \ddots & \vdots \\
	h_{n-2} & h_{n-1} & -h_0 & \vdots & -h_{n-3}\\
	h_{n-1} & -h_{0} & -h_1 & \ldots & -h_{n-2}
\end{matrix}
\right)
\]
is symmetric. Hence, it leads to a $q$-symplectic matrix $M_{\bar{\rho}(\vb{h})}$ as in (\ref{eqn:q-symplectic}). To simplify notation, we set
\[
M_\vb{h}:=M_{\rho(\vb{h})}\quad \text{and} \quad \bar{M}_\vb{h}:=M_{\bar{\rho}(\vb{h})}.
\] 
In Section~\ref{sec:analysis_distance} we analyze the first successive minimum of $\cL(M_{\bar{\rho}(\vb{h})})$ for random $\vb{h}$. To so we use the following result from Conrad et al.~\cite{conrad2024good} (Lemma 3 and Corollary 1).
\begin{lemma}\label{lem:isomorphic_simple}
Let $\vb{h}, M_{\vb{h}}, \bar{M}_{\vb{h}}$ be as above. Then the lattices $\cL(M_{\vb{h}})$ and $\cL(\bar{M}_{\vb{h}})$ are isomorphic. In particular, their first successive minima are identical.
\end{lemma}
\begin{proof}
Consider the $2n\times 2n$ matrix $U=\sigma_n\oplus I_n$. Then $U^T=U$ and $U^{-1}=U$. Hence, 
\begin{enumerate}
	\item $U$ is orthogonal,
	\item $U$ is unimodular.
\end{enumerate}
Furthermore, $U\cdot M_{\vb{h}}\cdot U= \bar{M}_{\vb{h}} $. This shows the first statement. The second statement follows from the fact that isomorphic lattices have the same successive minima.
\end{proof}
\begin{corollary}
	Let $\vb{h}, \rho(\vb{h}), M_{\vb{h}}$ be as above. Then there exists a $q$-symplectic lattice isomorphic to $\cL(M_{\vb{h}})$.
\end{corollary}
Next, we describe $\cL(M_{\rho(\vb{h})})$ as a module lattice. Later, this allows us to analyze its minimal distance, and, by consequence of Lemma~\ref{lem:isomorphic_simple}, the minimal distance of $q$-symplectic lattice $\cL(M_{\bar{\rho}(\vb{h})})$. Fix $\vb{h}\in R$ and $q\in \N$. Set 
\[
\vb{A}(\vb{h}):=\left(
\begin{matrix}
	\vb{h} & -\vb{1}
\end{matrix}
\right)\in R^{1\times 2}
\]
Consider $\latperp(\vb{A}(\vb{h}))$. Assume that $\vb{h}\in R_q^\times$, i.e. $\vb{h}$. Viewed as an $R$-module, this lattice has as generating set, in fact as basis, the rows of matrix 
\[
G_\vb{h}:=\left(
\begin{matrix}
	\vb{1} & \vb{h}\\
	\vb{0}& \vb{q}
\end{matrix}
\right).
\] 
Consequently, viewed as a $\Z$-module, or a lattice, $\latperp(\vb{A}(\vb{h}))$ has a generating set the rows of matrix $M_\vb{h}$, that is
\begin{equation}
	\latperp(\vb{A}(\vb{h}))=\cL(M_\vb{h})
\end{equation} 
Simplifying notation, we set
\[
\latperp(\vb{h}):=\cL(M_\vb{h}).
\]
Since $\cL(M_\vb{h})$ is isomorphic to the $\cL(\bar{M}_\vb{h})$, lattice $\latperp_\vb{h}$ is isomorphic to a $q$-symplectic lattice, namely lattice $\cL(\bar{M}_\vb{h})$. Accordingly,
$\frac{1}{\sqrt{q}}\latperp_\vb{h}$ is isomorphic to the symplectic lattice $\frac{1}{\sqrt{q}}\cL(\bar{M}_\vb{h})$. By abuse of notation, we call both, $\frac{1}{\sqrt{q}}\latperp_\vb{h}$ and $\frac{1}{\sqrt{q}}\cL(\bar{M}_\vb{h})$, \emph{symplectic $\rsis$ lattices}.  To construct symplectic $\rsis$ lattices with large minimal distance, we show that for random $\vb{h}\in \ringq$, lattice $\qfactor\latperp(\vb{h})$ has large minimal distance. For this, it turns out, the $R$-module view of $\latperp(\vb{h})$ is more convenient.
\subsection{Symplectic lattices from $\msis$ lattices}
We generalize the construction from above to matrices over $R$. The situation studied above, corresponds to matrices with a single row.  First some notation. Let $\vb{H}=(\vb{h}_{ij})_{1\le i,j\le k}$ be a $k\times k$ matrix over $R$. Then we denote by $\rho(\vb{H})$ the following  $nk\times nk$ over $\Z$
\begin{eqnarray*}
	\rho(\vb{H}):=\left(A(\vb{h}_{ij})\right) _{1\le i\le k},
\end{eqnarray*}
that is, we replace entries $\vb{h}_{ij}$ by the matrix representations of multiplication with $\vb{h}_{ij}.$ 

Consider a symmetric matrix $\vb{H}$ over $R$. Matrix $M_{\rho(\vb{H})}$ is not symmetric, since matrix representations of multiplications with ring elements are not symmetric. However, for symmetric $\vb{H}$ the matrix
\[
\bar{\rho}(\vb{H}):= \left(\sigma_n\cdot A(\vb{h}_{ij})\right) _{1\le i\le k}
\]
is symmetric.

Further, generalizing the arguments from above, consider the $2nk\times 2nk$ matrix 
\[
\vb{U}:=\underbrace{\sigma_n\oplus \cdots \oplus\sigma_n}_{\text{$k$ times}}\oplus\underbrace{I_n\oplus \cdots \oplus I_n}_{\text{$k$ times}},
\]
which generalizes the definition above from $k=1$ to arbitrary $k.$ Simple calculations show that
\begin{enumerate}
	\item $\vb{U}$ is orthogonal
	\item $\vb{U}$ is unimodular
	\item
	$
	\vb{U}\cdot M_{\rho(\vb{H})}\cdot \vb{U}=M_{\bar{\rho}(\vb{H})}.
	$
\end{enumerate}
Generalizing the simplifying notation from before, we set
\[
M_\vb{H}:=M_{\rho(\vb{H})}\quad \text{and} \quad \bar{M}_\vb{H}:=M_{\bar{\rho}(\vb{H})}.
\]
This leads to the following generalization of Lemma~\ref{lem:isomorphic_simple} to symmetric matrices.
\begin{lemma}\label{lem:isomorphic_general}
	Let $\vb{H}\in R^{k\times k}$ symmetric. Define $M_\vb{H}, \bar{M}_\vb{H} $ be as above. Then the lattices $\cL(M_\vb{H})$ and $\cL(\bar{M}_\vb{H})$ are isomorphic. In particular, their first successive minima are identical.
\end{lemma}
As before this yields
\begin{corollary}\label{cor:isomorphic_general}
Let $\vb{H}, \rho(\vb{h}), M_\vb{H}$ be as above. Then there exists a $q$-symplectic lattice isomorphic to $\cL(M_\vb{H})$	
\end{corollary}
 Generalizing further, set 
\[
\vb{A}(\vb{H}):=\left(
\begin{matrix}
	\vb{H} & -\vb{I}_k
\end{matrix}
\right)\in R^{k\times 2k}
\]
Consider $\latperp(\vb{A}(\vb{H}))$. 
Viewed as an $R$-module, this lattice has as generating set, in fact as basis, the rows of matrix 
\[
G_\vb{H}:=\left(
\begin{matrix}
	\vb{I}_k & \vb{H}\\
	\vb{0}_k& q\vb{I}_k
\end{matrix}
\right).
\] 
Consequently, viewed as a $\Z$-module, or a lattice, $\latperp(\vb{A}(\vb{H}))$ has as basis the rows of matrix $M_\vb{H}$, that is
\begin{equation}
	\latperp(\vb{A}(\vb{H}))=\cL(M_\vb{H}).
\end{equation} 
Simplifying notation, we define
\[
\latperp(\vb{H}):=\latperp(\vb{A}(\vb{H})).
\]
By Lemma~\ref{lem:isomorphic_general}, $\qfactor\latperp(\vb{H})$ is isomorphic to the lattice $\qfactor\cL(\bar{M}_\vb{H})$, which by construction is symplectic. As in the $\rsis$ case we abuse notation and call $\qfactor\latperp(\vb{H})$ a \emph{symplectic $\msis$ lattice}. 
To construct symplectic $\msis$ lattices with large minimal distance, we show that for random symmetric $\vb{H}\in \ringq^{k\times m}$, lattice $\qfactor\latperp(\vb{H})$ has large minimal distance.

%% file: Ring_Analysis.tex
\section{Analysis of minimal distance for symplectic $\rsis$  and $\msis$ lattices}\label{sec:ring_analysis_distance}
In this section we prove that with probability close to $1$ lattices $\latperp(\vb{h})$ and $\latperp(\vb{H})$ have minimal distance $\approx \sqrt{qn/\pi e}$. Consequently, lattices $\qfactor\latperp(\vb{h})$ and $\qfactor\latperp(\vb{H})$ have minimal distance $\approx \sqrt{n/\pi e}$ with probability close to $1$. Here the probabilities are with respect to the uniform distribution over elements $\vb{h}$ in $R_q$ and the uniform distribution over symmetric matrices $\vb{H}$ in $R_q^{k\times k}, k\in \N.$ The first distribution being the special for $k=1$ of the second distribution. In the sequel we focus on the case of general $k$. But to understand the analysis, it may be useful to think of the special case $k=1$, as this simplifies the notation.
\begin{definition}
	For $R_q$ as above and $k\in \N$ denote the set of symmetric matrices in $R_q^{k\times k}$ by $\rmatsym{k}$.
	 Hence, $R_q=\rmatsym{1}$.
\end{definition}
The analysis of the minimal distance proceeds in two steps and is mostly identical to the strategy in Section~\ref{sec:analysis_distance}. As in this section, it uses standard arguments in lattice-based cryptography (see for example~\cite{micciancio2011geometry}). However, the ring structure of $R$ and $R_q$ complicates the analysis significantly. In the first step of the analysis, we determine for fixed $\vb{z}\in \Z^{2kn},$ interpreted as an element in $R^{2k}$, the probability
\begin{equation}\label{eqn:prob_element_inv_lattice_ring}
	\Pr[\vb{z}\in \latperp(\vb{H}):\vb{H}\leftarrow \rmatsym{k}]
\end{equation}
This probability depends on the structure of the representative of $\vb{s}$ in $R_q\times R_q$ (see Equation~\ref{eqn:ring_isomorphism}). For each relevant structure, we derive an upper bound on the number of \emph{short} $\vb{z}$ with this structure.  Then, the main result on the minimal distance of lattices $\latperp(\vb{H})$ follows by the union bound. 
\subsection{Analyzing the probability}
Let $\Phi(X)=X^n+1, R, R_q$ as before. Fix an integer $k$. For an element $\vb{z}\in R^{2k}$, write $\vb{z}=(\vb{z}_1,\vb{z}_2), \vb{z}_i\in R$ and $\vb{z}_i=(\vb{z}_{i1},\ldots, \vb{z}_{ik}), \vb{z}_{ij}\in R$. As before, via its coefficient representation, we can identify an element $\vb{z}$ in $\R^{2k}$ with an element $\vb{z}\in \Z^{2kn}$. We use the same notation for both viewpoints.  For $\vb{H}\in R_q^{k\times k},$
\[
\latperp(\vb{H})=\left\{\vb{z}\in R^{2k}:\vb{H}\vb{z}_1=\vb{z}_2\mod qR\right\}.
\]
To analyze the probability in~(\ref{eqn:prob_element_inv_lattice}), we need to further distinguish elements in $\R^{2k}.$ 
Let 
\begin{equation*}
	X^n+1=\prod_{\substack{m\mid 2n \\ m\nmid n}}\prod_{j_m\in\setn{n_m}} p_{m,j_m}(X),
\end{equation*}
be the factorization of $X^n+1$ into irreducible factors modulo $q$ (see~(\ref{eq:factorizatzion_xn}). As before, $n_m$ denotes the number of factors of the $m$-th cyclotomic polynomial $Q_m(X)$ modulo $q$ and $d_m$ denotes their degree. The set of index pairs $(m,j_m)$ of divisors of $X^n+1$ modulo $q$ is denoted by $\cI$, 
\[
\cI:=\bigcup_{\substack{m\mid 2n \\ m\nmid n}}\{m\}\times \setn{n_m}.
\]
For $\vb{z}\in R^{2k}, \vb{z}\ne \vb{0} \mod qR,$  set
\begin{equation}\label{eqn:setT}
\csetT{m}{\vb{z}}:=\{j\in\setn{n_m}:\vb{z}_1=\vb{z}_2=0\mod p_{m,j}(X)\}
\end{equation}
and
\begin{equation}
	\setT{m}{\vb{z}}:=\setn{n_m}\setminus \csetT{m}{\vb{z}}, \quad T(\vb{z}):=\bigcup_{\substack{m\mid 2n \\ m\nmid n}}\{m\}\times \setT{m}{\vb{z}}, \quad \bar{T}(\vb{z}):=\cI\setminus T(\vb{z}).
\end{equation}
Next set 
\[
\ell_m({\vb{z}}):=|\setT{m}{\vb{z}}| \quad\text{and}\quad 
p_{\vb{z}}(X):=\prod_{(m,j_m)\in T(\vb{z})} p_{m,j_m}(X).
\]
Consequently, 
\begin{equation}
	d_{\vb{z}}:=|T(\vb{z})|=\sum_{\substack{m\mid 2n \\ m\nmid n}}\ell_m(\vb{z})d_m \quad \text{and}\quad \deg(p_{\vb{z}}(X))=d_{\vb{z}}.
\end{equation}
\begin{lemma}\label{lem:prob_element_H_lattice}
Let $\vb{z}\in R^{2k}, \vb{z}\ne \vb{0} \mod q R$. Then 
\[
\Pr[\vb{z}\in \latperp(\vb{H}); \vb{H}\leftarrow \rmatsym{k})=\frac{1}{q^{d_\vb{z}k}}.
\]	
\end{lemma}
\begin{proof}
\[
\vb{z}\in \latperp(\vb{H})\Leftrightarrow \vb{H}\vb{z}_1=\vb{z}_2 \mod qR \Leftrightarrow \forall m,j_m\;\;\vb{H}\vb{z}_1=\vb{z}_2 \mod p_{m,j_m}(X).
\]
Since the polynomials $p_{m,j_m}(X)$ are irreducible modulo $q$, for each pair $(m,j_m)$ the equation  $\vb{H}\vb{z}_1=\vb{z}_2\mod p_{m,j_m}(X)$ is an equation over the finite field $\fq{q}{d_m}$.
For $(m,j_m)\in \bar{T}(\vb{z})$ the equation is satisfied for all $\vb{H}.$ Hence, let $(m,j_m)\in T(\vb{z})$.  If $\vb{z}_1=\vb{0} \mod p_{m,j_m}(X),$ then $\vb{z}_2=\vb{0}\mod p_{m,j_m}(X),$ which is impossible by definition of $\setT{m}{\vb{z}}$. By Lemma~\ref{lem:solutions_symmetric_matrix} 
 \[
 \Pr[\vb{H}\vb{z}_1=\vb{z}_2 \mod p_{m,j_m}(X); \vb{H}\leftarrow \Fq[k\times k]{q^{d_m}} \;\text{symmetric}]=\frac{1}{q^{d_m}}.
 \]
 For distinct pairs  $(m,j_m)$ these events are independent. The lemma follows from the definition of $\ell_m(\vb{z})$ and $d_{\vb{z}}$.
 \end{proof}
 
 \subsection{Bounds on the number of short vectors}
 Let $\cT\subset \cI$.  In this subsection we derive bounds on the number of short elements in $\vb{z}=(\vb{z}_1,\vb{z}_2)\in R^{2k}$ satisfying $
 \vb{z}_1=\vb{z}_2=0\mod p_{m,j_m}(X)$ exactly for the pairs in $\cT$.
 To state this more precisely, for $\cT\subseteq \cI, \cT\ne \emptyset$, set 
 \begin{itemize}
 	\item $p_\cT(X):=\prod_{(m,j_m)\in \cT}p_{m,j_m}(X),\quad \bar{p}_\cT(X):= \prod_{(m,j_m)\not\in T}p_{m,j_m}(X)$
 	\item $d_\cT:=\deg(p_\cT(X)), \quad\bar{d}_\cT:=\deg(\bar{p}_\cT(X)) $
 	\item $g_\cT:=\gcd(n,d_\cT)$.
  \end{itemize}
 and
 \begin{equation}\label{eq:def_zt}
 Z_\cT:=\left\{\vb{z}\in R^{2k}: \bar{T}(\vb{z})=\cT\right\}.
 \end{equation}
\begin{lemma}\label{lem:general_bound_set_size}
	Let $0<r$ be such that $\radius<q/2$. Then
\begin{equation}\label{eqn:main_estimate_size}
\left|B_{2kn}(\radius)\cap Z_{\cT}\right|\le \frac{n}{g_\cT}\sigma_{2kd_\cT}\left(\frac{d_\cT}{n}\right)^{kd_\cT}\left(r+\sqrt{\frac{kn}{2q}}\right)^{2kd_\cT}q^{kd_\cT}
\end{equation}
\end{lemma}
\begin{proof}
	To simplify notation, we fix set $\cT$ and drop all subscripts with $\cT$, i.e. we use $Z,d,\bar{d},$ and $g$ instead of $Z_\cT, d_\cT,\bar{d}_\cT, g_\cT$.
	
Let  $\vb{z}\in B_{2kn}(\radius)\cap Z_{\cT} , \vb{z}=(\vb{z}_{11},\ldots,\vb{z}_{1k},\vb{z}_{21},\ldots,\vb{z}_{2k})\in R^{2k}$ with $\vb{z}_{ij}\in R, 1\le i\le k, j=1,2$.  Set
\[
\vbh:=X^{hg}\cdot \vb{z}_{ij} \mod X^n+1, h=0,\ldots, n/g-1
\]
The coefficient vector of $\vbh,$ also denoted by $\vbh,$ is obtained by cyclically shifting the coefficient vector by $hg$ positions and multiplying the first $gh$ coefficients by $-1$.  Each $\vbh$ can be written as 
\[
\vbh=\vbhu X^{\bar{d}}+\vbhl, \deg(\vbhu)<d=n-\bar{d}, \deg(\vbhl)<\bar{d},
\]
for elements $\vbhu,\vbhl\in R$. The elements $\vbh$ have the following properties.
\begin{enumerate}
	\item $\|\vbh\|^2=\|\vb{z}_{ij}\|^2$ for all indices $i,j,h$,
	\item up to sign, each coefficient of $\vb{z}_{ij}$ appears in exactly $d/g$ elements $\vbhu$,
	\item modulo $qR$, each element $\vbh$ is uniquely determined by $\vbhu$.
\end{enumerate}
Properties 1. and 2. are immediate. To verify property 3., note that $X^{\bar{d}}$ is invertible modulo all polynomials $p_{m,j_m}(X)$. Hence, $\vb{z}_{ij}=0\mod p_{m,j_m}(X)$ implies $\vbh=0\mod p_{m,j_m}(X)$ for all $(m,j_m)\not\in T$. Moreover, 
\[
\vbhl=X^{-\bar{d}} \vbhu \mod p_{m,j_m}(X), (m,j_m)\not \in \cT,
\]
where $X^{-\bar{d}}$ denotes the multiplicative inverse of $X^{\bar{d}}$ modulo $p_{m,j_m}(X)$. Since the $p_{m,j_m}(X)$ are pairwise co-prime, their product is $\bar{p}(X)$, and $\deg(\bar{p}(X))=\bar{d}>\deg(\vbhl)$, property 3. follows from the Chinese Remainder Theorem.

By assumption $\radius<q/2$ and,  therefore, every element $\vb{b}\in B_{2kn}(\radius)$ must consist of polynomials in $R$, whose coefficients lie between $-q/2$ and $q/2$. That is, we only need to consider elements in $Z\cap R_q^{2k}$ with overall length bounded by $\radius$. Here we take the integers between $-q/2$ and $q/2$ as a set of representatives for $\Z_q$. Next by property 1. and 2.
\[
\sum_{h=0}^{n/g-1}\sum_{i=1}^k\sum_{j=1}^2 \|\vbhu\|^2 < \frac{d}{g} r^2q.
\]
Hence, there exists a least one index $t\in\{0,\ldots, n/g-1\}$ such that
\[
\sum_{i=1}^k\sum_{j=1}^2 \|\bar{\vb{z}}_{ij}^{(t)}\|^2<\frac{d}{n}r^2q.
\]
By property 3. elements $\bar{\vb{z}}_{ij}^{(t)}, 1\le i\le k, j=1,2$ uniquely determine $\vb{z}$ and setting 
\[
\bar{\vb{z}}^{(h)}=(\bar{\vb{z}}_{11}^{(h)}, \ldots, \bar{\vb{z}}_{1k}^{(h)},\bar{\vb{z}}_{21}^{(h)}, \ldots, \bar{\vb{z}}_{2k}^{(h)}),
\]
for every $\vb{z}\in Z\cap B_{2kn}(\radius)$ there is an $t, 0\le t\le n/g-1,$ and an element
$\bar{\vb{z}}^{(t)}\in B_{2dk}(\dradius)$ that uniquely determines it. 

Therefore
\[
|Z\cap B_{2kn}(\radius)|\le \frac{n}{g}|\Z^{2dk}\cap B_{2dk}(\dradius)|.
\]
By Lemma~\ref{lem:bound_integer_points},
\begin{align*}
	|\Z^{2dk}\cap B_{2dk}(\dradius)| &\le  \sigma_{2dk} (\dradius+\sqrt{2dk}/2)^{2dk}\\
&= \sigma_{2kd}\left(\frac{d}{n}\right)^{kd}\left(r+\sqrt{\frac{kn}{2q}}\right)^{2kd}q^{kd}
\end{align*}	
This proves the lemma.
\end{proof}
\subsection{Probability analysis}\label{sec_ring_analysis}

\begin{theorem}\label{thm:main_analysis}
	With the notation from above
\begin{multline}\label{eq:thm_prob_estimate}
\Pr[\lambda_1(\latperp(\vb{H}))<r\sqrt{q}; \vb{H}\leftarrow \rmatsym{k}] \le \\
	\sum_{\cT\subseteq \cI, \cT\ne \emptyset}\frac{n}{g_\cT}\sigma_{2kd_\cT}\left(\frac{d_\cT}{n}\right)^{kd_\cT}\left(r+\sqrt{\frac{kn}{2q}}\right)^{2kd_\cT}.
\end{multline}
\end{theorem}
To simplify the notation, denote the right-hand side of the above inequality by
\[
\epsilon(n,q,k):=\sum_{\cT\subseteq \cI,\cT\ne \emptyset}\frac{n}{g_\cT}\sigma_{2kd_\cT}\left(\frac{d_\cT}{n}\right)^{kd_\cT}\left(r+\sqrt{\frac{kn}{2q}}\right)^{2kd_\cT}.
\]
\begin{proof}
The proof is a simple combination of Lemma~\ref{lem:prob_element_H_lattice} and Lemma~\ref{lem:general_bound_set_size}. To simplify notation, define event
\[
E:=\{\vb{H}\in \rmatsym{k}:\lambda_1(\latperp(\vb{H}))<\radius\},
\] 
and denote by $\bar{E}$ its complimentary event. 
Hence, it suffices to prove that
\[
\Pr[ E(\vb{H};\vb{H}\leftarrow \rmatsym{k}]\
\]

By the union bound
	\begin{equation*} 
	\begin{split} 
				\Pr[\vb{H}\in E; \vb{H}\leftarrow \rmatsym{k}] & = \Pr[\exists \vb{s}\in B_{2n}(\radius)\cap \latperp(\vb{H})\setminus \{0\}; \vb{H}\leftarrow \rmatsym{k}]\\
		 & \le \sum_{\vb{s}\in B_{2kn}(\radius)\cap \Z^{2kn}\setminus\{0\}}\Pr[\vb{s}\in \latperp(\vb{H}); \vb{H}\leftarrow \rmatsym{k}].
	\end{split}
	\end{equation*}
We split the sum into $Z_\cT$ corresponding to subsets $\cT\subset \cI$. 
\begin{multline}\label{eqn:main_prob_next_to_last}
	\sum_{\vb{s}\in B_{2kn}(\radius)\cap \Z^{2kn}\setminus\{0\}}\Pr[\vb{s}\in \latperp(\vb{H}); \vb{H}\leftarrow \rmatsym{k}] = \\
	\sum_{\cT\subset \cI}\sum_{\vb{s}\in Z_\cT\cap B_{2kn}(\radius)}\Pr[\vb{s}\in \latperp(\vb{H}); \vb{H}\leftarrow \rmatsym{k}].
\end{multline}
By Lemma~\ref{lem:prob_element_H_lattice}, for $\vb{s}\in Z_\cT$, 
\[
\Pr[\vb{s}\in \latperp(\vb{H}); \vb{H}\leftarrow \rmatsym{k}]\le \frac{1}{q^{kd_\cT}},
\]
and by Lemma~\ref{lem:general_bound_set_size}, for $\cT\subset \cI$,
\[
|Z_\cT\cap B_{2kn}(\radius)|\le \frac{n}{g_\cT}\sigma_{2kd_\cT}\left(\frac{d_\cT}{n}\right)^{kd_\cT}\left(r+\sqrt{\frac{kn}{2q}}\right)^{2kd_\cT}q^{kd_\cT}.
\]
Together with (\ref{eqn:main_prob_next_to_last}) this proves the theorem.
\end{proof}
We rewrite the probability estimate above in a form that is better suited to further analysis. To do so, for two sequences  $\vec{j}=(j_1,\ldots,j_s), \vec{n}=(n_1,\ldots,n_s)$ of natural numbers, we set $\vec{j}\le \vec{n}$ iff for all $i$ we have $j_i\le n_i$. If $\vec{j}$ is not the all zero vector we write $\vec\ne \vec{0}$. For $\vec{j}\le \vec{n}$, define
\[
{\vec{n}\choose \vec{j}}:=\prod_{i=1}^s {n_i \choose j_i}.
\]
From now on $\vec{n}=(n_m)_{ m\mid 2n, m\nmid n}$ is the sequence of number of factors in the factorization modulo $q$ of the cyclotomic polynomials dividing $X^n+1$. Furthermore $(d_m)_{ m\mid 2n, m\nmid n}$ is the sequence of degrees of these factors.  For $\vec{j}\le \vec{n}$ we set
\[
d_{\vec{j}}:=\sum_{i=1}^s j_i d_i
\]
and
\[
{\vec{n}\choose \vec{j}}:=\prod_{i=1}^s {n_i \choose j_i}.
\]
Finally, set
\[
\alpha:= r+\sqrt{\frac{kn}{2q}}
\]

\begin{lemma}
With the notation above 
	\begin{multline}\label{eq:cor_prob_estimate_basic}
\Pr[\lambda_1(\latperp(\vb{H}))<r\sqrt{q}; \vb{H}\leftarrow \rmatsym{k}] \le \\
n\sum_{\vec{j}\le \vec{n},\vec{j}\ne \vec{0}}{\vec{n}\choose \vec{j}}\sigma_{2kd_{\vec{j}}}
\left(\frac{d_{\vec{j}}}{n}\right)^{kd_{\vec{j}}}\alpha^{2kd_{\vec{j}}}.
\end{multline}
\end{lemma}
\begin{proof}
For each set $\cT$ in (\ref{eq:thm_prob_estimate}), only $g_\cT$ and $d_\cT$ depend on $\cT$. Clearly, $g_\cT$ is lower bounded by $1$. Furthermore, all $\cT$ that for all $m$ contain the same number $j_m$ of elements from $\{m\}\times \setn{n_m}$ have the same degree $d_{\vec{j}}, \vec{j}=(j_m)_{m\mid 2n, m\nmid n}$. 
\end{proof}

\begin{lemma}\label{lem:ring_analysis_main}
	With the notation from above and with $d:=\min\{d_m\}$, 
\begin{multline}\label{eq:cor_prob_estimate_binomial}
\Pr[\lambda_1(\latperp(\vb{H}))<r\sqrt{q}; \vb{H}\leftarrow \rmatsym{k}] \le \\
	 \frac{n}{\sqrt{2\pi kd}}\left(\left(1+\left(\frac{e\pi\alpha^2}{kn}\right)^{dk}\right)^{\sum n_m}-1\right)
\end{multline}
\end{lemma}
\begin{proof}
For every integer $d$, $\sigma_{2d}=\frac{\pi^d}{d!}$. Using Stirling's formula (Lemma~\ref{lem:bounds}), we obtain for every vector $\vec{j}$ 
\[
\sigma_{2kd_{\vec{j}}} \left(\frac{d_{\vec{j}}}{n}\right)^{kd_{\vec{j}}}\le \frac{1}{\sqrt{2\pi d_{\vec{j}} k}}\left(\frac{\pi e}{d_{\vec{j}} k}\right)^{kd_{\vec{j}}}.
\]
Next, we can bound the right-hand side in (\ref{eq:cor_prob_estimate_basic}) as follows
\begin{align*}
n\sum_{\vec{j}\le \vec{n},\vec{j}\ne \vec{0}}{\vec{n}\choose \vec{j}}\sigma_{2kd_{\vec{j}}}
\left(\frac{d_{\vec{j}}}{n}\right)^{kd_{\vec{j}}}\alpha^{2kd_{\vec{j}}} \le &
n\sum_{\vec{j}\le \vec{n},\vec{j}\ne \vec{0}}{\vec{n}\choose \vec{j}}\frac{1}{\sqrt{2\pi d_{\vec{j}} k}}
\left(\frac{e\pi}{nk}\right)^{kd_{\vec{j}}}\alpha^{2kd_{\vec{j}}}\\
\le & \frac{n}{\sqrt{2\pi kd}}\sum_{\vec{j}\le \vec{n},\vec{j}\ne \vec{0}}{\vec{n}\choose \vec{j}}
\left(\frac{e\pi}{nk}\right)^{kd_{\vec{j}}}\alpha^{2kd_{\vec{j}}}\\
= &\frac{n}{\sqrt{2\pi kd}}\left(\prod_{m}\left(1+\left(\frac{e\pi \alpha^2}{kn}\right)^{kd_m}\right)^{n_m}-1\right)\\
\le &\frac{n}{\sqrt{2\pi kd}}\left(\left(1+\left(\frac{e\pi\alpha^2}{kn}\right)^{kd}\right)^{\sum n_m}-1\right)
\end{align*}

\end{proof}
\begin{theorem}\label{thm:relation_prob_length}
With the notation from above and for every $0<p\le 1, k\le n$, if 
\[
r:=\sqrt{\frac{kn}{e\pi}}\left(\frac{p\sqrt{kd}}{n\sum n_m}\right)^{\frac{1}{2kd}}-\sqrt{\frac{nk}{2q}}
\]
then 
\[
\Pr[\lambda_1(\latperp(\vb{H}))<r\sqrt{q}; \vb{H}\leftarrow \rmatsym{k}]\le p. 
\]
\end{theorem}
\begin{proof} With our choice of $r$, we have 
\[
\alpha:=\sqrt{\frac{kn}{e\pi}}\left(\frac{p\sqrt{kd}}{n\sum n_m}\right)^{\frac{1}{2kd}}.
\]
Then we can bound the right-hand side in (\ref{eq:cor_prob_estimate_binomial}) as follows, using 
 $(1+\gamma/m)^m\le 1+2\gamma,$ valid for $\gamma\le 1/2$, and $p\sqrt{kd}/n\le 1/2$, by choice of $p,k$, 
\begin{align*}
	\frac{n}{\sqrt{2\pi kd}}\left(\left(1+\left(\frac{e\pi\alpha^2}{kn}\right)^{kd}\right)^{\sum n_m}-1\right) =& \frac{n}{\sqrt{2\pi kd}}\left(\left(1+\frac{p\sqrt{kd}}{n\sum n_m}\right)^{\sum n_m}-1\right)\\
	\le & \frac{n}{\sqrt{2\pi kd}}\frac{2p\sqrt{kd}}{n}\\
	\le & p.
\end{align*}
\end{proof}
\subsection{The case $n=2^e$}
As we are interested not only in asymptotic results but also in exact results with precise constants, and because it is most relevant for the connection of GKP codes to the NTRU cryptosystem, in this section we consider the case that $n$ is a power of two, $n=2^e$, and prime $q=5\mod 8$ or $q=3\mod 16$, in more detail.  

Let $n=2^e$ be a power of $2$.  In this case, over the rationals $X^n+1=Q_{2m}(X)$ is irreducible (see (\ref{eqn:factors_xn+1}). The proof of Theorem 2' in Chapter 4 of \cite{ireland2013classical} shows that for $q=5\mod 8, n\ge 8$, the order of $q$ modulo $2n$ is $n/2$. The same proof can be used for $q=3\mod 16$ and $n\ge 16$. By  Theorem 2.47 in~\cite{lidl1994introduction} (see also (\ref{eq:factorizatzion_xn}), 
\begin{equation}\label{eq:factorizatzion_xn_power_of_2}
X^n+1=p_{2n,1}(X)p_{2n,2}(X),
\end{equation}
where polynomials $p_{2n,j}(X)$ have degree $n/2$. Using the notation from above we have $\cI=\{(2n,1),(2n,2)\}$, and the possible subsets $\cT$ in Theorem~\ref{thm:main_analysis} are 
\[
\cT=\{(2n,1),(2n,2)\}, \{(2n,1)\}, \{(2n,2)\}.
\]
The corresponding values for $d_\cT, g_\cT,$ and $n/\gcd(n,g_\cT)$ are
\[
d_\cT=g_\cT=n,n/2,n/2; \quad n/\gcd(n,g_\cT)=1,2,2.
\]
Taking this into account, as well as $\sum_m n_m=2, d=\min\{d_m\}=n/2$, Lemma~\ref{lem:ring_analysis_main} can be improved to
\begin{lemma}\label{lem:ring_analysis_power_of_2}
Let $n=2^e$ and $q=5\mod 16$ or $q=3\mod 8$, then 
\begin{multline}\label{eq:cor_prob_estimate_binomial_power_of_2}
\Pr[\lambda_1(\latperp(\vb{H}))<r\sqrt{q}; \vb{H}\leftarrow \rmatsym{k}] \le \\
	 \frac{2}{\sqrt{\pi kn}}\left(\left(1+\left(\frac{e\pi\alpha^2}{kn}\right)^{kn/2}\right)^{2}-1\right)
\end{multline}	
\end{lemma}
Using this lemma and some straightforward calculation, we obtain
\begin{theorem}\label{thm:relation_prob_length_power_of_two}
Let $n=2^e$ and $q=5\mod 8$ or $q=3\mod 16$. For every $0<p\le \frac{6}{\sqrt{kn}}$, if 
\[
r:=\sqrt{\frac{kn}{e\pi}}\left(\frac{p\sqrt{kn}}{4}\right)^{\frac{1}{kn}}-\sqrt{\frac{nk}{2q}}
\]
then 
\[
\Pr[\lambda_1(\latperp(\vb{H}))<r\sqrt{q}; \vb{H}\leftarrow \rmatsym{k}]\le p. 
\]
\end{theorem}

%% file: MainTheoremsRings.tex
\section{Main results - module and ring constructions}\label{sec:maintheoremsring}
In this section, based on the analysis of the previous section, we prove lower bounds on the minimum distance of random $q$-ary lattices, focussing on two types of bounds: a lower bound of $\sqrt{nk/\pi e}$ exactly and  lower bounds close to $\sqrt{nk/\pi e}$, that can be obtained with probability exponentially close to $1$. We begin with the strongest and most interesting case, namely $n=2^e$ being a power of $2$. 
\begin{theorem}\label{thm:ring_power_of_two}
Let $k,n,q\in \N$ with $n$ a power of $2$ and $q$ a prime with $q=5 \mod 8$ or $q=3\mod 16$. Set $R:=\Z[X]/(X^n+1)$ and $R_q:= \Z[q]/(X^n+1)$.
Choose $\vb{H}$ uniformly at random from $\rmatsym{k}$ and set
\[
\cL:=\frac{1}{\sqrt{q}}\latperp(\vb{H}).
\]
\begin{enumerate}
	\item For $q\ge 11\pi e (nk)^2$, with probability at least $1-\frac{6}{\sqrt{nk}}$ 
	\[
	\lambda_1(\cL)\ge  \sqrt{\frac{nk}{\pi e}},
	\]
	\item For $q\ge 2\pi e(nk)^{3/2}$, with probability at least $1-e^{-(nk)^{1/4}/2}$
		\[
		\lambda_1(\cL)\ge 
	\left(1-\frac{1}{(nk)^{3/4}}\right)\sqrt{\frac{nk}{\pi e}}.
	\]

\end{enumerate}
\end{theorem}
\begin{proof}
Note that $\lambda_1(\cL)=\frac{1}{\sqrt{q}}\lambda_1(\latperp(\vb{H}))$. Hence for both statements it suffices the prove the statements for $\latperp(\vb{H})$ with bounds $\sqrt{q}\lambda_1(\cL)$.

To prove the first statement, set $p:=6/\sqrt{nk}$. Then
\[
\left(\frac{p\sqrt{kn}}{4}\right)^{\frac{1}{nk}}=(3/2)^{\frac{1}{nk}}\ge 1+\frac{\ln(3/2)}{nk}.
\]
With our choice of $q$ it follows that
\[
\sqrt{\frac{nk}{2q}}\le \frac{\ln(3/2)}{nk}\sqrt{\frac{nk}{\pi e}}.
\]
	The statement follows from Theorem~\ref{thm:relation_prob_length_power_of_two}.

To prove the second statement we proceed similarly. 
First, with our choice of $p$ we obtain
\[
\left(\frac{p\sqrt{kn}}{4}\right)^{\frac{1}{nk}}\ge e^{(nk)^{3/4}/2}\ge \left(1-\frac{1}{2(nk)^{3/4}}\right).
\]
By choice of $q$
\[
\sqrt{\frac{nk}{2q}}\le \frac{1}{2(nk)^{3/4}}.
\]
As before, the statement follows from Theorem~\ref{thm:relation_prob_length_power_of_two}.
\end{proof}
For the general case of $n,q$ and the factorization of $X^n+1$ modulo $q$, a lower bound on the minimum distance of a random $q$-ary lattice can achieved with high probability if $d$ is large enough.
\begin{theorem}\label{thm:ring_general}
Let $k,n,q\in \N$ with $q$ prime and let $d$ be the minimum degree of a polynomial in the factorization of $X^n+1$ modulo $q$. Let $0<p,\gamma <1$. If 
\[
d\ge \frac{\ln(1/p)+2\ln(n)}{2\gamma k}\quad \text{and} \quad q\ge \frac{2\pi e}{\gamma^2}
\]
then for $\vb{H}$ uniformly at random from $\rmatsym{k}$ and $
\cL:=\frac{1}{\sqrt{q}}\latperp(\vb{H})$ with probability at least $1-p$
\[
\lambda_1(\cL)\ge (1-\gamma)\sqrt{\frac{nk}{\pi e}}.
\]	
\end{theorem}
\begin{proof}
	We proceed as in the proof of the previous theorem, using Theorem~\ref{thm:relation_prob_length} instead of Theorem~\ref{thm:ring_power_of_two}.
	From the lower bound on $d$, it follows that
\[
\left(\frac{p\sqrt{kd}}{n\sum_m n_m} \right)^{\frac{1}{2kd}}\ge e^{-\gamma/2}\ge (1-\gamma/2).
\]
The bound on $q$ implies
\[
\frac{nk}{2q}\le \gamma/2\frac{nk}{\pi e}.
\]
The theorem follows from Theorem~\ref{thm:relation_prob_length}. 
\end{proof}
It is not difficult to show that, using Theorem~\ref{thm:relation_prob_length}, the bound in the previous theorem cannot be improved significantly.
\paragraph{Instantiations}
First, we show how to calculate $d:=\min\{d_m\}$ given specific values of $n$ and $q$. To do so, we need to recall some facts from elementary number theory, that yield more details on the factorization of $X^n+1$ modulo $q$, (see (\ref{eq:factorizatzion_xn})). 

We use the \emph{Carmichael function $\lambda(\cdot)$}, where for $n\in \N$, $\lambda(n)$ is the smallest integer $m$ such that
\[
\forall a \in \Z_n^*: a^m=1\mod n.
\]
Equivalently, $\lambda(n)$ is the largest integer $m$ such $\Z_n^*$ has an element $a$ of order $m$. An element of order $\lambda(n)$ is often called a \emph{primitive $\lambda$-root modulo $n$}. 
We later use two well-known properties of the $\lambda$-function and primitive $\lambda$-roots modulo $n$.
The first is a recursive formula for $\lambda(n)$.
\begin{equation}\label{eq:carmichael_formula}
\lambda(n)=
\begin{cases}
	\phi(n) & \text{if $n$ is $2,4$ or an odd prime power $p^e$}\\
	\frac{1}{2}\phi(n) & \text{if $n=2^e, e\ge 3$}\\
	\lcm(\lambda(n_1),\ldots,\lambda(n_s)) & \text{if $n=n_1\cdots n_s$ for distinct prime powers $n_i$}
\end{cases}	
\end{equation}
The second property follows from the recursive formula and standard properties of the rings $\Z_n^*$ (see for example~\cite{ireland2013classical}).
\begin{lemma}~\label{lem:primitive_root_divisor}
	Let $n,d\in \N$ with $d$ a divisor of $n$. If $a$ is a primitive $\lambda$-root modulo $n$, then $a\mod d$ is a primitive $\lambda$-root modulo $d$.
\end{lemma}
Finally, by the prime number theorem for arithmetic progressions (see for example~\cite{apostol2013introduction}) , for every $n$ there also exists a prime $q$ that is a primitive $\lambda$-root modulo $n$.
These theorems and facts together lead to the following theorem that will be used for the construction of symplectic lattices and GKP codes from $\msis$ lattices.
\begin{theorem}\label{thm:factorization_complete_modulo_q}
	Let $n,q\in \N$ with $q$ a prime that is a primitive $\lambda$-root modulo $2n$. Then the polynomial $X^n+1$ modulo $q$ factors as
	\[
	X^n+1=\prod_{\substack{d\mid 2n \\ d\nmid n}}\prod_{j=1}^{\phi(d)/\lambda(d)} Q_{d,j}(X),
	\]
	where the $Q_{d,j}(X)$ are the irreducible factors of $Q_d(X)$ modulo $q$, and each $Q_{d,j}(X)$ has degree $\lambda(d)$.
\end{theorem}
This theorem implies for a primitive $\lambda$-root $q$ of $n$:
\[
d:=\min\{d_m\}:=\min\{\lambda(m):m\mid 2n \land m\nmid n\}
\]
Write $n=2^e n',$ with $n'$ odd. Then this characterization of $d$ together with (\ref{eq:carmichael_formula}) implies
\[
d=
\begin{cases}
	2^e & e=0,1\\
	2^{e-1}& e\ge 2.
\end{cases}
\]
If $n=2^e$ is a power of $2$ and $q=5\mod 8$ or $q=3\mod 16$, then $q$ is a primitive $\lambda$-root modulo $q$. Hence, as already used in Theorem~\ref{thm:ring_power_of_two}, $d=n/2$.

Theorem~\ref{thm:ring_general} can give meaningful results even for $d=1$. In particular, if we exploit the general module structure, i.e. we may choose $k$ freely, although we want to keep it as small as possible. As an example, in Theorem~\ref{thm:ring_general}, let us choose $d=1$ and $p=1/n$. Then we get 
\[
k\ge 3\ln(n)/(2\gamma).\]
Hence by increasing $k$ we can get, with probability close to $1$, GKP codes whose minimum distance gets closer to $\sqrt{nk/\pi e}$.

If on the other hand, we want to choose $k=1$, i.e. we choose ring lattices, not general module lattice, then with $p=1/n$ Theorem~\ref{thm:ring_general} implies
\[
d\ge \frac{3\ln(n)}{2\gamma}, 
\]
hence $d$ must be a least logarithmic.

One example where $d=1$, in fact, where all factors of $X^n+1$ modulo $q$ are linear, occurs when $q=1\mod n$. This case is quite popular in lattice-based cryptography, since it allows for fast arithmetic in $R_q$ using the so called \emph{number theoretic transformation (NTT)}, a variant of the Fourier transform. As we will see in the next section, the case $q=1\mod n$ also can lead to efficient decoding algorithms for GKP codes.

%% file: Ring_and_NTRU.tex
\section{GKP codes from $\ntru$ lattices}\label{sec:ntru_gkp}
In this section we combine the results from Section~\ref{sec:ring_symplectic_construction} with results from \cite{stehle2011making} to show that NTRU lattices can be used to construct good GKP codes, thereby proving a conjecture by Conrad et al.~\cite{conrad2024good}. 

For the exact definition of the $\ntru$ cryptosystem we refer to~\cite{hoffstein1999ntru}. By $\distntru{q,\sigma}$ we refer to the following distribution, which we call  the $q$-ary $\ntru$ distribution with parameter $\sigma$ or simply the $\ntru$ distribution. First, denote by $D_\sigma^\times$ the discrete Gaussian distribution $D_{\Z^{n},\sigma}$ restricted to $R_q^\times$ (via rejection sampling). Then $
	\distntru{q,\sigma}:=\frac{D_1}{D_2} \mod q,
$
where $D_1, D_2$ are independently distributed as the Gaussian distribution $D^\times_\sigma$
Note that $\distntru{q,\sigma}$ is a distribution on $R_q^\times$. By $U(R_q^\times)$ and $U(R_q)$ we denote the uniform distributions on $R_q^\times$ and $R_q$, respectively.

For the remainder of this section let $n=2^e$ be a power of two and $q=5\mod 8$ or $q=3\mod 16$. In this case $X^n+1$ is irreducible over $\Q$ and it factors into two polynomials of degree $n/2$ over $\fq{q}{}$. From this it follows using the Chinese remainder theorem that
\begin{eqnarray}\label{eq:size_R_invertible}
\left|R_q^\times \right|=(q^{n/2}-1)^2	\quad \text{and} \quad \left|R_q\setminus R_q^\times \right|<2q^{n/2}.
\end{eqnarray} 
For two distributions $D_1, D_2$ with identical finite support $T$, denote by $\Delta(D_1,D_2)$ their statistical distance, that is
\[
\Delta(D_1,D_2):=\max\{|D_1(S)-D_2(S)|\mid S\subseteq T\}=\frac{1}{2}\sum_{t\in T}|D_1(t)-D_2(t)|.
\]
From (\ref{eq:size_R_invertible}) one deduces
\begin{equation}\label{eq:dist_R_R_invertible}
	\Delta(U(R_q),U(R_q^\times))\le \frac{2q^{n/2}}{q^n-2q^{n/2}}\le\frac{3}{q^{n/2}}, \quad \text{$q>5$}.
\end{equation}
As a special case of Theorem 3.2 in \cite{stehle2011making}, we obtain the following theorem.
\begin{theorem}\label{thm:stehle_steinfeld}[Stehl\'{e}, Steinfeld] Let $n=2^e$ be a power of two and $q=5\mod 8$ or $q=3\mod 16$. Let $0<\epsilon<1/4$. Then
\[
\Delta\left(\distntru{q,\sigma},U(R_q^\times)\right)\le 2^{10n}q^{-\epsilon n}, \; \text{if}\; \sigma\ge \sqrt{n\ln(8n)}q^{1/2+\epsilon}\; \text{and}\; q\ge n^{\frac{2}{1-4\epsilon}}.
\]	
\end{theorem}
Combining this theorem with (\ref{eq:dist_R_R_invertible}) yields
\begin{corollary}\label{cor:dist_R_NTRU}
	Let $n=2^e, n>5,$ be a power of two and $q=5\mod 8$ or $q=3\mod 16$. Let $0<\epsilon<1/3$. Then
\[
\Delta\left(\distntru{q,\sigma},U(R_q)\right)\le (2^{10n}+3)q^{-\epsilon n}, \; \text{if}\; \sigma\ge \sqrt{n\ln(8n)}q^{1/2+\epsilon}\; \text{and}\; q\ge n^{\frac{2}{1-4\epsilon}}.
\]	
\end{corollary}

For $k=1$, we have that $\rmatsym{k}=U(R_q)$. Hence, combining Corollary~\ref{cor:dist_R_NTRU} with part (2) of Theorem~\ref{thm:ring_power_of_two} leads to the following theorem, proving Conjecture 2 from ~\cite{conrad2024good}, at least asymptotically.
\begin{theorem}
	Let $n=2^e, n>5,$ be a power of two and $q=5\mod 8$ or $q=3\mod 16$. Let $0<\epsilon<1/4$. Choose $\vb{h}$ from the NTRU distribution $\distntru{q,\sigma}$, where $\sigma\ge \sqrt{n\ln(8n)}q^{1/2+\epsilon}$ and $q\ge \max\{11\pi 2 n^{3/2}, n^{\frac{2}{1-4\epsilon}}\}$.   Then
\[
\lambda_1(\cL)\ge 
	\left(1-\frac{1}{(n)^{3/4}}\right)\sqrt{\frac{n}{\pi e}},
\]
except with probability $e^{-n^{1/4}/2}+ (2^{10n}+3)q^{-\epsilon n}$.

\end{theorem}

%% file: Ring_BDD.tex
\section{Bounded distance decoding in the module case and the role of $k$}\label{sec:ring_bdd}
We can also use algorithm $\bddtriv$ from Section~\ref{sec:bdd} to decode the GKP codes obtained from $\msis$ lattices. The most time consuming step in algorithm $\bddtriv$ is the matrix-vector multiplication $\vb{c'}\vb{H}$ in step 2, which requires $n^2$ multiplications in $\F_q$. If $n$ is a power of $2$, instead of the $\sis$ lattice construction of GKP codes we can use the $\msis$-based construction using Theorem~\ref{thm:ring_power_of_two}, with $k=1$. The matrix-vector multiplication in step $2$ of algorithm $\bddtriv$ becomes a multiplication of two ring elements in $R=Z[X](X^n+1)$. By construction, this ring contains $n$-th roots of unity which allows us to perform multiplication in $R$ with $\cO(n\log(n))$ arithmetic operations in $\F_q$ using the \emph{fast Fourier transform (FFT)}. 

However, powers of $2$ are relatively sparse. Using $k>1$ allows us to obtain fast decoding for more general $n$. If we choose the dimension of the code as $n\cdot k$ with $n$ being a power of $2$, the multiplication $\vb{c'}\vb{H}$ requires $k^2$ multiplications in ring $R=\Z[X]/(X^n+1)$ at the cost $\cO(k^2n\log(n))$ arithmetic operations in $F_q$, again using FFT. Hence, the parameter $k$ and the use of $\msis$ lattices and not just $\rsis$ lattice (corresponding to $k=1$) allows us more flexibility in designing code GKP codes with efficient decoding algorithms. Using other choices for $n$ than just powers of $2$, yields yet more flexibility in designing good GKP codes for various dimensions.

%% file: Ring_Simulation.tex
\section{Simulation of main results for ring constructions}\label{sec:ring_simulation}
In this section we show some simulation results in Section \ref{sec:maintheoremsring}, i.e. the minimal distance for symplectic R-SIS and M-SIS lattices. The decoding performance for these lattices shows no difference from SIS lattices in Section \ref{sec:simulation}, so we omit those results here.

\begin{figure}[htbp]
  \centering

  \begin{subfigure}[b]{0.49\textwidth}
    \centering
    \includegraphics[width=\linewidth]{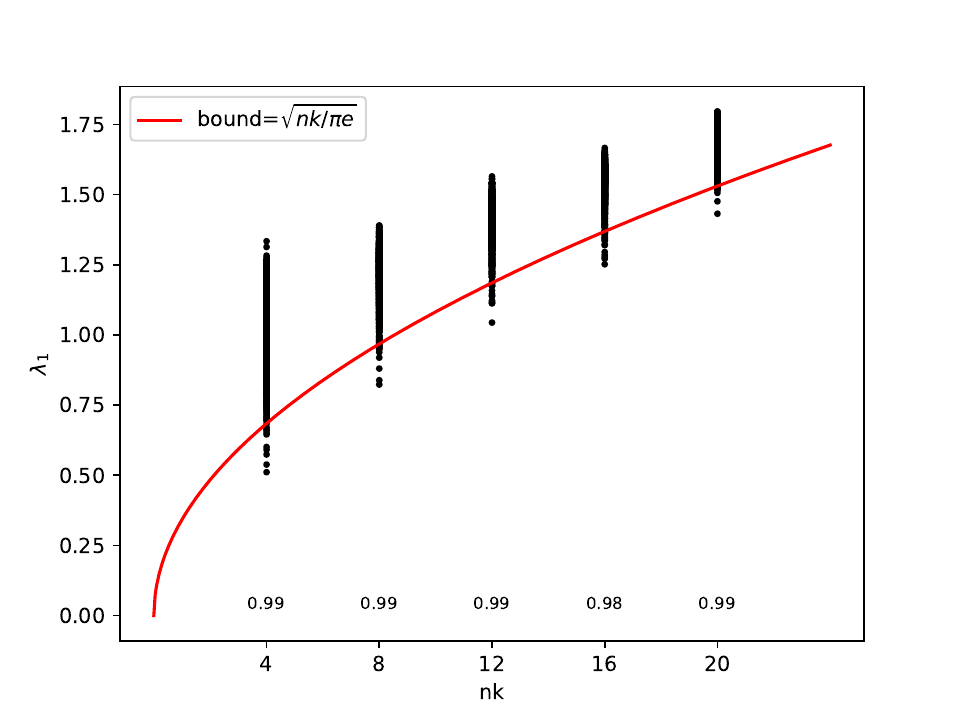}
    \caption{$n=4$, $m\in [1,5]$, $q=701$}
    \label{ring_sim_a}
  \end{subfigure}
  \hfill
  \begin{subfigure}[b]{0.49\textwidth}
    \centering
    \includegraphics[width=\linewidth]{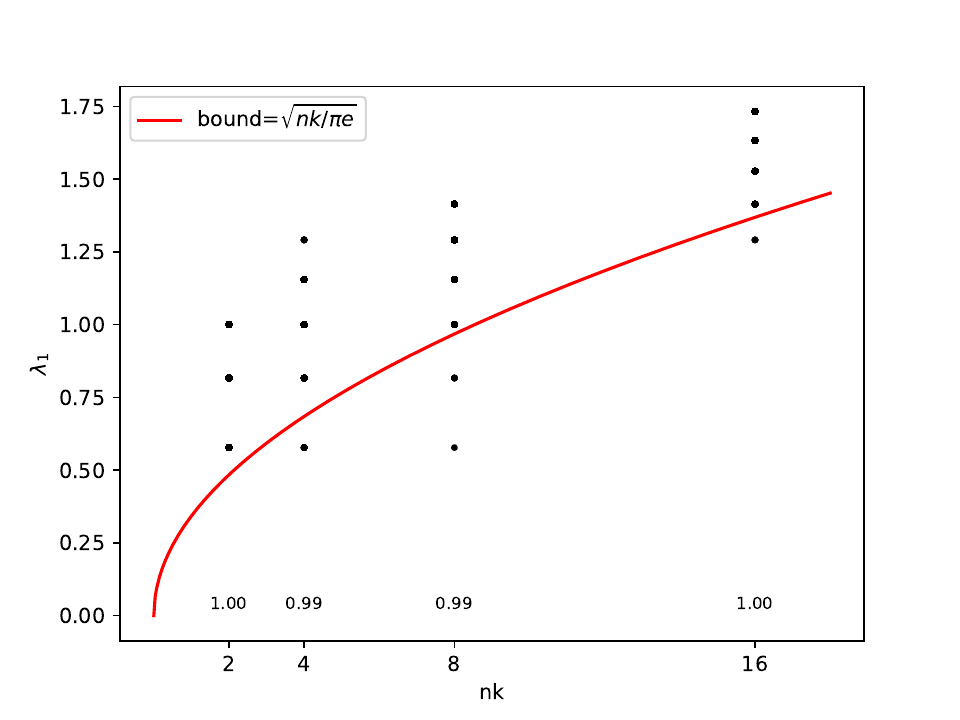}
    \caption{$n=2^e$, $m=1$, $q=3$}
    \label{ring_sim_b}
  \end{subfigure}

  \vspace{0.3cm}

  \begin{subfigure}[b]{0.49\textwidth}
    \centering
    \includegraphics[width=\linewidth]{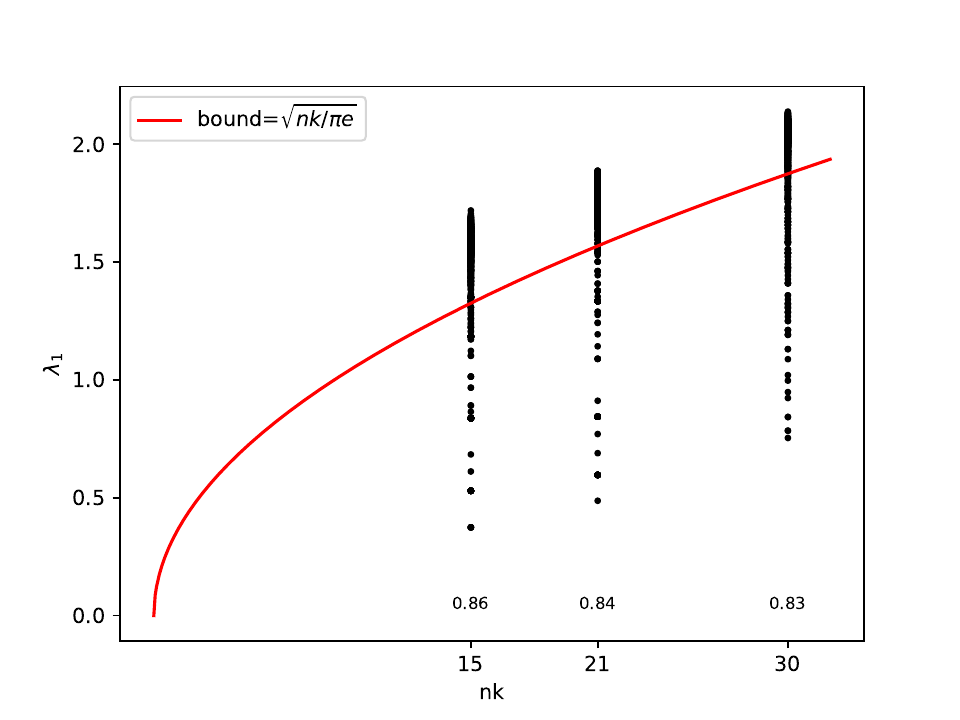}
    \caption{$n=2^ep_1p_2$}
    \label{ring_sim_c}
  \end{subfigure}
  \hfill
  \begin{subfigure}[b]{0.49\textwidth}
    \centering
    \includegraphics[width=\linewidth]{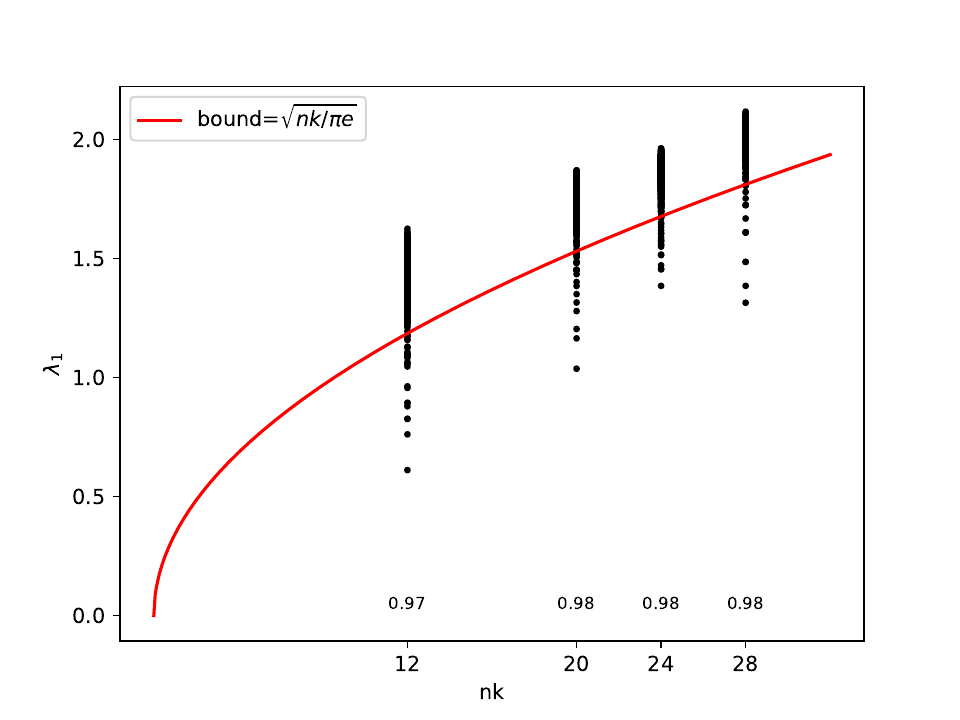}
    \caption{$n=2^ep$}
    \label{ring_sim_d}
  \end{subfigure}

  \caption{Simulations of minimal distance for symplectic R-SIS and M-SIS lattices.  (a) $n=4$, $m\in [1,5]$, $q=701$. (b) $n=2^e$, $m=1$, $q=3$. (c) $n=2^ep_1p_2$. (d) $n=2^ep$.}
  \label{ring_sim}
\end{figure}

Figure \ref{ring_sim} shows the simulation results of minimal distance for symplectic R-SIS and M-SIS lattices. For each combination of $(n,k,q)$, 1000 lattice samples are generated with $\vb{H}\leftarrow \rmatsym{k}$. Minimal distance of lattice is then calculated via BKZ reduction. In the figure, the red line represents the asymptotic bound $\sqrt{nk/\pi e}$ and the black dots represent the minimal distance for each lattice sample. The numbers at the bottom show, for each combination $(n,k)$, how many samples have minimal distance $\lambda_1>\sqrt{nk/\pi e}$. Figure \ref{ring_sim_a} corresponds to the case of Theorem \ref{thm:ring_power_of_two}, where $n=4$ is a power of 2, and $q=701$ satisfies $q=5 \mod 8$ and $q\geq \frac{\pi e (4\cdot 5)^2}{2}$. As shown in the figure, in all combinations of $(n,k)$, the generated M-SIS lattices have minimal distance larger than $\sqrt{nk/\pi e}$ with probability close to 1, which is even much larger than what the theorem tells; Figure \ref{ring_sim_b} also has $n$ equal to a power of 2, but a small $q=3$ is used in the simulation. Still the generated R-SIS lattices have minimal distance larger than $\sqrt{nk/\pi e}$ with probability close to 1. Thus from the simulation we make the conjecture that the generated M-SIS lattices have a high probability of achieving $\lambda_1\geq \sqrt{nk/\pi e}$ regardless of what $q$ is used for their generation. Figures \ref{ring_sim_c} and \ref{ring_sim_d} illustrate more general instances covered by Theorem \ref{thm:ring_general}. In these examples, $n$ is not a power of 2, and the observed probabilities are lower than in the power-of-two case, which is consistent with the more delicate factorization behavior of $X^n+1$ modulo $q$. Again the results correspond well with the theorem, and even suggest a higher probability of $\lambda_1\geq \sqrt{nk/\pi e}$.

%% file: Conclusion.tex
\section{ Conclusion and Outlook}\label{sec:conclusion}

In this paper, we give the first efficient constructions of good families of GKP codes. Our constructions achieve the code distances non-constructively shown
to be achievable by Preskill and Harrington~\cite{Harrington-Preskill_2001} (based on results on symplectic lattices by Buser and Sarnak~\cite{buser1994period}) and conjectured by Conrad et al. to be achievable constructively by NTRU-based GKP codes~\cite{conrad2024good}. Furthermore, the achieved distances are provably optimal or near optimal (up to small constants). Our GKP codes are constructed from well-known families of cryptographic lattices, general $\sis$ lattices, ring and module $\sis$ lattices. Combining our results on GKP codes from $\rsis$ lattices with results from Stehl\'{e} and Steinfeld~\cite{stehle2011making}, we are able to prove Conjecture 2 from~\cite{conrad2024good}, namely,  for certain parameter choices the NTRU-based construction of GKP codes yields GKP codes with distance $\sqrt{n/\pi e}$.  We also showed a trivial decoding algorithm for our GKP codes which does not use trapdoors. Experiments confirm the effectiveness of the constructions. For parameter choices and decoding algorithms the experiments even supersede the theoretical results. 

Our results raise various interesting research questions. The NTRU constructions in~\cite{conrad2024good} use trapdoors to obtain improved decoding algorithms for the NTRU-based GKP codes. The $\sis$ lattices that we use in our constructions do not have trapdoors. Hence the questions arises whether we can use other cryptographic lattices with trapdoors to construct good GKP codes with efficient decoding algorithms exploiting the trapdoors. Examples of such lattices include \emph{learning with error ($\lwe$)} lattices. Also lattices with so called $G$-trapdoors may be of interest (for both see ~\cite{peikert2016decade}). Our results for $\msis$ lattices, though general, lead to non-trivial results only for limited parameter settings. We conjecture that techniques developed for the ring-$\lwe$  and module-$\lwe$ problems (see~\cite{lyubashevsky2013toolkit}) can be used to obtain improved results applicable to more general parameter settings. Finally, and practically most importantly, GKP codes as presented in this paper need to be combined with other quantum error-correcting codes like surface codes or CSS-codes to obtain efficient codes applicable for the design of large scale quantum computers. Hence, what are the best and most efficient combinations of GKP and other quantum error-correcting codes?

For randomized GKP codes constructed from cryptographic lattices, one can also think about exploiting their randomized and cryptographic property. For example, Conrad et al.~\cite{conrad2024good} proposed a quantum communication protocol with NTRU-based GKP codes. The protocol is said to be secure with the trapdoors used in NTRU systems. Similar protocol should also work for other code constructions that have trapdoors such as $\lwe$. For GKP codes that do not have trapdoors in their construction, it may also be possible to use them for constructing secure quantum communication protocols by using the randomness in the codes. The randomness of these codes and their relations with cryptographic lattices may also lead to possible decoding algorithms which are more efficient or better-performed compared with normal GKP codes.

%% file: Acknowledgement.tex
\section*{Acknowledgments}

This work was funded by the Ministry of Culture and Science of the State of North Rhine-Westphalia and by the German Federal Ministry of Research, Technology and Space (BMFTR) within the PhoQuant project (Grant No. 13N16103). Additionally, this work was supported by the Equal Opportunity Program, Grant Line 2: Support for Female Junior Professors and Postdocs through Academic Staff Positions, 15th funding round of Paderborn University.